\documentclass[pra,
aps,
twocolumn,
reprint,
showpacs,
groupedaddress,
floatfix, 
superscriptaddress,
longbibliography]{revtex4-1}

\pdfoutput=1

\usepackage{general_latex}
\usepackage[justification=centerfirst]{caption}
\usepackage{dsfont}

\usepackage{amsmath}
\usepackage{amssymb}
\usepackage{amsthm}

\newtheoremstyle{mystyle}
  {}                                      
  {}                                      
  {\itshape}                              
  {}                                      
  {\bfseries}                             
  {.}                                     
  { }                                     
  {\thmname{#1}\thmnumber{ #2}\thmnote{ (#3)}}%

\theoremstyle{mystyle}
\newtheorem{theorem}{Theorem}
\newtheorem{definition}[theorem]{Definition} 
\newtheorem{corollary}[theorem]{Corollary} 
\newtheorem{lemma}[theorem]{Lemma}

\newcommand{\nth}{\bar{n}_{\scriptsize\textrm{th}}}
\newcommand{\nthtiny}{\bar{n}_{\tiny\textrm{th}}}

\newcommand{\reg}{\scriptsize\textrm{reg}}
\newcommand{\sq}{\scriptsize\textrm{sq}}
\newcommand{\hex}{\scriptsize\textrm{hex}}
\newcommand{\sqtiny}{\tiny\textrm{sq}}
\newcommand{\hextiny}{\tiny\textrm{hex}}

\begin{document}
\title{
Improved quantum capacity bounds of Gaussian loss channels and achievable rates with Gottesman-Kitaev-Preskill codes
} 
\author{Kyungjoo Noh}\email{kyungjoo.noh@yale.edu}
\affiliation{Department of Applied Physics and Physics, Yale University, New Haven, Connecticut 06520, USA}
\affiliation{Yale Quantum Institute, Yale University, New Haven, Connecticut 06520, USA}
\author{Victor V. Albert}
\affiliation{Walter Burke Institute for Theoretical Physics and Institute for Quantum Information and Matter, California Institute of Technology, Pasadena, California 91125, USA}
\author{Liang Jiang}
\affiliation{Department of Applied Physics and Physics, Yale University, New Haven, Connecticut 06520, USA}
\affiliation{Yale Quantum Institute, Yale University, New Haven, Connecticut 06520, USA}
\begin{abstract}
Gaussian loss channels are of particular importance since they model realistic optical communication channels. Except for special cases, quantum capacity of Gaussian loss channels is not yet known completely. In this paper, we provide improved upper bounds of Gaussian loss channel capacity, both in the energy-constrained and unconstrained scenarios. We briefly review the Gottesman-Kitaev-Preskill (GKP) codes and discuss their experimental implementation. We then prove, in the energy-unconstrained case, that the GKP codes achieve the quantum capacity of Gaussian loss channels up to at most a constant gap from the improved upper bound. In the energy-constrained case, we formulate a biconvex encoding and decoding optimization problem to maximize the entanglement fidelity. The biconvex optimization is solved by an alternating semidefinite programming (SDP) method and we report that, starting from random initial codes, our numerical optimization yields GKP codes as the optimal encoding in a practically relevant regime.
\end{abstract}
\maketitle

\section{Introduction}

Quantum communication is an important area of quantum technology wherein classical communication is enriched with non-local quantum entanglement and quantum state transmission \cite{Wilde2013}. Similarly as in the classical case, practical quantum communication channels are inevitably noisy, and thus the quantum information should be encoded and decoded in a non-trivial way such that the effects of channel noise can be reversed by quantum error correction \cite{Shor1995}. 

Quantum capacity of a channel quantifies the maximum amount of quantum bits per channel use that can be transmitted faithfully (upon optimal encoding and decoding) in the limit of infinite channel uses. Analogous to mutual information in the classical communication theory, \textit{regularized coherent information} of a channel characterizes the channel's quantum capacity \cite{Schumacher1996,Lloyd1997,Devetak2005}. Unlike its classical counterpart, however, coherent information might be \textit{superadditive}, which makes it hard to evaluate the channel's quantum capacity \cite{DiVincenzo1998,Smith2008,Smith2011,Cubitt2015}.   

Bosonic Gaussian channels \cite{Eisert2005,Weedbrook2012,Serafini2017} are among the most studied quantum channels due to its relevance to optical communication. The bosonic pure-loss channel is a special case of Gaussian loss channels.
Since coherent information of the bosonic pure-loss channel is additive, its quantum capacity is known \cite{Holevo2001,Wolf2007,Wilde2012} (see \cite{Wilde2016} for the general formalism of energy-constrained quantum capacity,
and our Theorem \ref{theorem:thermal optimizer is enough for bosonic pure loss channel capacity} 
for a pedagogical self-contained derivation of energy-constrained quantum capacity of bosonic pure-loss channels). For general Gaussian loss channels with added thermal noise, a lower bound of quantum capacity can be obtained by evaluating one-shot coherent information of the channel \cite{Holevo2001}, and several upper bounds are obtained by using, e.g., data-processing inequality \cite{Holevo2001,Pirandola2017,Sharma2017,Rosati2018}.

Evaluation of the quantum capacity, however, does not lend explicit encoding and decoding strategies achieving the capacity. In parallel with the characterization of quantum capacity, many bosonic quantum error-correcting codes have also been developed over the past two decades, using a few coherent states \cite{Cochrane1999,Niset2008,Leghtas2013,Mirrahimi2014,Lacerda2016,Albert2018}, position/momentum eigenstates \cite{Lloyd1998,Braunstein1998,Gottesman2001,Menicucci2014,Hayden2016,Ketterer2016}, finite superpositions of the Fock states \cite{Chuang1997,Knill2001,Ralph2005,Wasilewski2007,Bergmann2016a,Michael2016,
Yuezhen2017} of the bosonic modes. Hybrid CV-DV schemes have also been proposed \cite{Lee2013,Kapit2016}. 
In the context of communication theory, the Gottesman-Kitaev-Preskill (GKP) codes \cite{Gottesman2001} are particularly interesting as they can achieve one-shot coherent information of the Gaussian random displacement channel \cite{Harrington2001}. In our recent work, we showed, both numerically and analytically, that the GKP codes also exhibit excellent performance against the bosonic pure-loss errors, outperforming many the other bosonic codes \cite{Albert2017}.
It is important to understand why GKP code can perform so well, which will guide us to search for better bosonic codes.

In this paper, we prove that the GKP codes achieve quantum capacity of Gaussian loss channels up to at most a constant gap from an upper  bound of the quantum capacity. In section \ref{section:Preliminaries}, we review and summarize key properties of Gaussian loss, amplification and random displacement channels, which will be used in later sections. In section \ref{section:Quantum capacity of Gaussian channels}, we provide an improved upper bound of quantum capacity of the Gaussian loss channel, both in energy-constrained and unconstrained cases by introducing a slight modification of an earlier result in \cite{Sharma2017} (Theorems \ref{theorem:improved upper bound of quantum capacity energy unconstrained}, \ref{theorem:improved upper bound of quantum capacity energy constrained} and \ref{theorem:optimized data-processing upper bound}; see also Eq.\eqref{eq:improved upper bound of Gaussian random displacement channel} for an improved upper bound of Gaussian random displacement channel capacity). In section \ref{section:GKP codes}, we discuss experimental implementation of the GKP codes and prove that the achievable quantum communication rate with the GKP codes deviates only by a constant number of qubits per channel use from our improved upper bound, assuming an energy-unconstrained scenario (Theorem \ref{theorem:achievable rate of the GKP codes for Gaussian loss channels}). 
In section \ref{section:Biconvex encoding and decoding optimization}, we address the energy-constrained encoding scenario and formulate a biconvex optimization problem to find an optimal set of encoding and decoding which maximize the entanglement fidelity. We solve the biconvex optimization using alternating semidefinite programming (SDP), and show that the GKP code emerges as an optimal solution. 

\section{Preliminaries}
\label{section:Preliminaries}

In this section, we provide a summary of preliminary facts about Gaussian channels which will be referenced in later sections. For introduction to bosonic mode, Gaussian state and Gaussian unitary operations, we refer the readers to \cite{Weedbrook2012,Serafini2017} (and Appendix \ref{section:Bosonic mode, Gaussian states and Gaussian unitary operations}, which is a part of \cite{Weedbrook2012} relevant to this paper translated into our notation). Lemmas \ref{lemma:loss plus amplification is displacement pre-amplification} and \ref{lemma:Gaussian loss channel decomposed into pure loss and amplficiation pre-amplification} are the key facts which will be used to derive our two main results in subsection \ref{subsection:GKP code rate} (Theorem \ref{theorem:achievable rate of the GKP codes for Gaussian loss channels}) and subsection \ref{subsection:Improved upper bound of Gaussian loss channel capacity} (Theorems \ref{theorem:improved upper bound of quantum capacity energy unconstrained} and \ref{theorem:improved upper bound of quantum capacity energy constrained}), respectively. 

\subsection{Gaussian channels}
\label{subsection:Gaussian channels}

Let $\mathcal{H}^{\otimes N}$ be a Hilbert space associated with $N$ bosonic modes, $\mathcal{L}(\mathcal{H}^{\otimes N})$ the space of linear operators acting on $\mathcal{H}^{\otimes N}$ and $\mathcal{D}(\mathcal{H}^{\otimes N})\equiv \lbrace \hat{\rho}\in\mathcal{L}(\mathcal{H}^{\otimes N})\, |\, \hat{\rho}^{\dagger}=\hat{\rho}\succeq 0,\textrm{Tr}[\hat{\rho}]=1\rbrace$ the space of density operators. A quantum channel $\mathcal{N} : \mathcal{L}(\mathcal{H}^{\otimes N})\rightarrow \mathcal{L}(\mathcal{H}^{\otimes N})$ maps a quantum state $\hat{\rho}\in\mathcal{D}(\mathcal{H}^{\otimes N})$ to another one in $\mathcal{D}(\mathcal{H}^{\otimes N})$ via a completely positive and trace-preserving (CPTP) map \cite{Choi1975}. Gaussian channels map a Gaussian states (see Eq.\eqref{eq:Wigner characteristic function for Gaussian states} for the definition) to another Gaussian state, and can be simulated by $\mathcal{N}(\hat{\rho}) = \textrm{Tr}_{E}[\hat{U}_{G}(\hat{\rho}\otimes \hat{\rho}_{E})\hat{U}_{G}^{\dagger}]$, where $\hat{U}_{G}$ is a Gaussian unitary operation on the system plus environmental modes, $\hat{\rho}_{E}$ is a Gaussian state and $\textrm{Tr}_{E}$ is partial trace with respect to the environment. Let $\mathbf{\hat{X}}^{T}=(\mathbf{\hat{x}}^{T},\mathbf{\hat{y}}^{T})$ be a collection of quadrature operators of the system mode $\mathbf{\hat{x}}$ and the environmental mode $\mathbf{\hat{y}}$ (see the text below Eq.\eqref{eq:coherent state equals to displacec vacuum} for the definition of quadratures), and assume that the initial system and environmental states are given by $\hat{\rho}_{G}(\mathbf{\bar{x}},\mathbf{V_{x}})$ and $\hat{\rho}_{G}(\mathbf{\bar{y}},\mathbf{V_{y}})$, respectively. Here, $\mathbf{\bar{x}}$, $\mathbf{\bar{y}}$ are the first moments and $\mathbf{V_{x}}$, $\mathbf{V_{y}}$ are the second moments of the system and environment--cf, Eq.\eqref{eq:Wigner characteristic function for Gaussian states}. If the Gaussian unitary on the joint system is characterized by
\begin{equation}
\mathbf{S} = \begin{pmatrix}
\mathbf{S_{xx}} & \mathbf{S_{xy}}   \\
\mathbf{S_{yx}} & \mathbf{S_{yy}}
\end{pmatrix} 
\,\,\textrm{and}\,\,
\mathbf{D}= \begin{pmatrix}
\mathbf{d_{x}}    \\
\mathbf{d_{y}}
\end{pmatrix},  \label{eq:Gaussian unitary characterization for Gaussian channels}
\end{equation}  
the first two moments of the system mode are transformed as $\mathbf{\bar{x}}\rightarrow \mathbf{S_{xx}}\mathbf{\bar{x}}+\mathbf{S_{xy}}\mathbf{\bar{y}}+\mathbf{d_{x}}$ and $\mathbf{V_{x}}\rightarrow \mathbf{S_{xx}}\mathbf{V_{x}}\mathbf{S^{T}_{xx}}+\mathbf{S_{xy}}\mathbf{V_{y}}\mathbf{S^{T}_{xy}}$, as can be derived by specializing Eq.\eqref{eq:transformation of moments of quadrature by Gaussian unitary} to Eq.\eqref{eq:Gaussian unitary characterization for Gaussian channels}. After tracing out the environment, the effective Gaussian channel for the system is characterized by 
\begin{equation}
\mathbf{\bar{x}}\rightarrow \mathbf{T}\mathbf{\bar{x}} + \mathbf{d}, \quad \mathbf{V_{x}}\rightarrow \mathbf{T}\mathbf{V_{x}}\mathbf{T}^{T}+\mathbf{N}, 
\end{equation}
where $\mathbf{T}=\mathbf{S_{xx}}$, $\mathbf{N}=\mathbf{S_{xy}}\mathbf{V_{y}}\mathbf{S^{T}_{xy}}$ and $\mathbf{d}=\mathbf{S_{xy}}\mathbf{\bar{y}}+\mathbf{d_{x}}$. 

The main subject of this paper is the Gaussian channel which models excitation loss, e.g., in optical communication channels. 
\begin{definition}[Gaussian loss channel]
Let $\hat{B}(\eta)$ be the beam splitter unitary with transmissivity $\eta\in[0,1]$, acting on mode $1$ and $2$. Then, a Gaussian loss channel is defined as 
\begin{equation}
\mathcal{N}[\eta,\nth](\hat{\rho}_{1}) \equiv \textrm{Tr}_{2}[\hat{B}(\eta)(\hat{\rho}_{1}\otimes \hat{\rho}_{\nthtiny})\hat{B}^{\dagger}(\eta)]. 
\end{equation}
Here, $\textrm{Tr}_{2}$ is the partial trace with respect to the mode $2$, initially in the thermal state $\hat{\rho}_{\nthtiny}$ with average photon number $\nth$. \label{definition:Gaussian loss channels} 
\end{definition} 

Note that $\mathcal{N}[\eta,\nth]$ is characterized by $\mathbf{T}=\sqrt{\eta}\mathbf{I}_{2}$, $\mathbf{N}=(1-\eta)(\nth+\frac{1}{2})\mathbf{I}_{2}$ and $\mathbf{d}=0$, where $\mathbf{I}_{n}$ is the $n\times n$ identity matrix, as can be derived from the definition of beam splitter unitary (see Eq.\eqref{eq:symplectic transformation beam splitter}). Bosonic pure-loss channel is a special case of Gaussian loss channel with $\nth=0$. Bosonic pure-loss channel with transmissivity $\eta\in[\frac{1}{2},1]$ ($\eta\in[0,\frac{1}{2}]$) is degradable (anti-degradable) \cite{Caruso2006} and its quantum capacity is known. (See section \ref{section:Quantum capacity of Gaussian channels} for the definition of degradability and anti-degradability.) Except for this special case, Gaussian loss channels are not degradable nor anti-degradable \cite{Caruso2006B,Holevo2007}, and thus their coherent information may not be additive. 

By replacing beam splitter unitary in Definition \ref{definition:Gaussian loss channels} with two-mode squeezing unitary, we obtain the Gaussian amplification channel. Quantum-limited amplification channel is a special case of Gaussian amplification and is defined as follows: 
\begin{definition}[Quantum-limited amplification channel] 
Let $\hat{S}_{2}(G)$ be the two-mode squeezing unitary operation with gain $G\ge 1$, acting on mode $1$ and $2$. Then, the quantum-limited amplification channel is defined as 
\begin{equation}
\mathcal{A}[G](\hat{\rho}_{1})  \equiv \textrm{Tr}_{2}[\hat{S}_{2}(G)(\hat{\rho}_{1}\otimes |0\rangle\langle 0|_{2} )\hat{S}_{2}^{\dagger}(G)], \label{eq:definition of quantum limited amplification}
\end{equation}
where $|0\rangle$ is the vacuum state. \label{definition:quantum limited amplification}
\end{definition}

$\mathcal{A}[G]$ is characterized by $\mathbf{T}=\sqrt{G}\mathbf{I}_{2}$, $\mathbf{N}=\frac{(G-1)}{2}\mathbf{I}_{2}$ and $\mathbf{d}=0$, as can be derived from the definition of two-mode squeezing (see Eq.\eqref{eq:symplectic transformation two mode squeezer}). Note that the noise $\mathbf{N}=\frac{(G-1)}{2}\mathbf{I}_{2}$ is due to the variance of the ancillary vacuum state, transferred to the system via two-mode squeezing. Since the vacuum has the minimum variance allowed by Heisenberg uncertainty principle, quantum-limited amplification incurs the least noise among all linear amplification channels \cite{Heffner1962}.  
  
Finally, we introduce the Gaussian random displacement channel, which GKP codes \cite{Gottesman2001} were originally designed to protect against. 
\begin{definition}[Gaussian random displacement channel]
Gaussian random displacement channel is defined as 
\begin{equation}
\mathcal{N}_{B_{2}}[\sigma^{2}](\hat{\rho}) \equiv \frac{1}{\pi\sigma^{2}} \int d^{2}\alpha e^{-\frac{|\alpha|^{2}}{\sigma^{2}}} \hat{D}(\alpha)\hat{\rho}\hat{D}^{\dagger}(\alpha),  
\end{equation}
where $\hat{D}(\alpha)$ is the displacement operator and $\sigma^{2}$ is the variance of random displacement.\label{definition:Gaussian displacement channel}
\end{definition}
We note that the Gaussian random displacement channel belongs to the class $B_{2}$ channel    
\cite{Caruso2006B,Holevo2007} and is characterized by $\mathbf{T}=\mathbf{I}_{2}$, $\mathbf{N}=\sigma^{2}\mathbf{I}_{2}$ and $\mathbf{d}=0$. (Since convention for the Gaussian random displacement channel varies in literature, we prove in Appendix \ref{subsection:Characterization of Gaussian random displacement channel} that $\mathbf{N}=\sigma^{2}\mathbf{I}_{2}$ holds in our convention, which is aligned with \cite{Gottesman2001,Harrington2001,Albert2017}).   

\subsection{Synthesis and decomposition of Gaussian channels}
\label{subsection:Synthesis and decomposition of Gaussian channels}

In Ref \cite{Albert2017}, it has been shown
that the GKP codes outperform many other bosonic codes in protecting encoded quantum information against bosonic pure-loss errors
upon optimal decoding operation numerically obtained by semidefinite programming. 
This excellent performance of the GKP codes can be achieved by a sub-optimal decoding operation, which is designed based on the observation that the bosonic pure-loss channel can be converted into the Gaussian displacement channel by a quantum-limited amplification: 
\begin{equation}
\mathcal{A}[1/\eta] \cdot \mathcal{N}[\eta,0] = \mathcal{N}_{B_{2}}[\sigma^{2}_{\eta,0}], \label{eq:pure loss plus amplification is displacement}
\end{equation} 
where $\sigma_{\eta,0}^{2} \equiv \frac{1-\eta}{\eta}$ (see Eq.(7.21) in \cite{Albert2017}). 
This simple fact was previously noted in \cite{Eisert2005} (but without proof and explicit relation between $\eta$ and $\sigma^{2}$) and has been used in a key disctribution scheme with the GKP codes \cite{Gottesman2001B}. Eq. \eqref{eq:pure loss plus amplification is displacement} implies that we can convert the loss channel into the displacement channel and then use the conventional GKP decoding (see the text below Eq.\eqref{eq:encoding one mode GKP as desired}) to decode the encoded GKP states corrupted by bosonic pure-loss channel. Here, we generalize this result and show that general Gaussian loss channels (with added thermal noise) can also be converted into Gaussian random displacement channel. Before doing so, we state and prove the following known fact:  
\begin{theorem}[Gaussian channel synthesis; Eq.(26) of \cite{Caruso2006B}]
Let $\mathcal{N}_{1}$ and $\mathcal{N}_{2}$ be Gaussian channels with specification $(\mathbf{T}_{1},\mathbf{N}_{1},\mathbf{d}_{1})$ and $(\mathbf{T}_{2},\mathbf{N}_{2},\mathbf{d}_{2})$, respectively. Then, the synthesized channel $\mathcal{N}\equiv \mathcal{N}_{2}\cdot\mathcal{N}_{1}$ is a Gaussian channel with specification 
\begin{align}
\mathbf{T} &= \mathbf{T}_{2}\mathbf{T}_{1},
\nonumber\\
\mathbf{N} &= \mathbf{T}_{2}\mathbf{N}_{1}\mathbf{T}_{2}^{T} + \mathbf{N}_{2},
\nonumber\\
\mathbf{d} &= \mathbf{T}_{2}\mathbf{d}_{1}+\mathbf{d}_{2}. 
\end{align} \label{theorem:Gaussian channel synthesis}
\end{theorem}
\begin{proof}
Let $\hat{\rho}$ be a Gaussian state, i.e., $\hat{\rho}=\hat{\rho}_{G}(\mathbf{\bar{x}},\mathbf{V})$. Upon $\mathcal{N}_{1}$, the first two moments of $\hat{\rho}$ are transformed into $\mathbf{\bar{x}}'=\mathbf{T}_{1}\mathbf{\bar{x}}+\mathbf{d}_{1}$ and $\mathbf{V}' = \mathbf{T}_{1}\mathbf{V}\mathbf{T}_{1}^{T} +\mathbf{N}_{1}$. The second channel $\mathcal{N}_{2}$ then transforms $\mathbf{\bar{x}}'$ and $\mathbf{V}'$ into 
\begin{align}
\mathbf{\bar{x}}'' &= \mathbf{T}_{2} \mathbf{\bar{x}}' +\mathbf{d}_{2},
\nonumber\\ 
&=\mathbf{T}_{2}\mathbf{T}_{1}\mathbf{\bar{x}} + \mathbf{T}_{2}\mathbf{d}_{1}+\mathbf{d}_{2}, 
\nonumber\\
\mathbf{V}''&= \mathbf{T}_{2}\mathbf{V}'\mathbf{T}_{2}^{T} + \mathbf{N}_{2} 
\nonumber\\
&=\mathbf{T}_{2}\mathbf{T}_{1}\mathbf{V}\mathbf{T}_{1}^{T}\mathbf{T}_{2}^{T} + \mathbf{T}_{2}\mathbf{N}_{1}\mathbf{T}_{2}^{T} + \mathbf{N}_{2}. 
\end{align}
The synthesized channel $\mathcal{N}=\mathcal{N}_{2}\cdot\mathcal{N}_{1}$ is thus a Gaussian channel mapping $\hat{\rho}_{G}(\mathbf{\bar{x}},\mathbf{V})$ to $\hat{\rho}_{G}(\mathbf{\bar{x}}'',\mathbf{V}'')$, with the specification as stated in the theorem. 
\end{proof}
Theorem \ref{theorem:Gaussian channel synthesis} then leads to the following generalization of Eq.\eqref{eq:pure loss plus amplification is displacement}. 
\begin{corollary}[Loss + amplification = displacement]
Let $\mathcal{N}[\eta,\nth]$ be the general Gaussian loss channel with transmissivity $\eta\in[0,1]$ and thermal photon $\nth$ in the ancillary mode. Let $\mathcal{A}[1/\eta]$ be the quantum-limited amplification with gain $G=1/\eta$. Then, 
\begin{equation}
\mathcal{A}[1/\eta] \cdot \mathcal{N}[\eta,\nth] = \mathcal{N}_{B_{2}}[\sigma^{2}_{\eta,\nthtiny}], 
\end{equation} 
where 
\begin{equation}
\sigma_{\eta,\nthtiny}^{2} \equiv \Big{(}\frac{1-\eta}{\eta}\Big{)}(\nth+1).  \label{eq:effective variance of displacement post-amp} 
\end{equation}  \label{corollary:loss plus amplification is displacement post-amplification}
\end{corollary} 
\begin{proof}
Note that $\mathcal{N}[\eta,\nth]$ and $\mathcal{A}(1/\eta)$ are characterized by $(\mathbf{T}_{1},\mathbf{N}_{1},\mathbf{d}_{1}) = (\sqrt{\eta}\mathbf{I}_{2},(1-\eta)(\nth +\frac{1}{2})\mathbf{I}_{2},\mathbf{0})$ and $(\mathbf{T}_{2},\mathbf{N}_{2},\mathbf{d}_{2}) = (\sqrt{\frac{1}{\eta}}\mathbf{I}_{2},\frac{1}{2}(\frac{1-\eta}{\eta})\mathbf{I}_{2},\mathbf{0})$, respectively. Let $(\mathbf{T},\mathbf{N},\mathbf{d})$ be the specification of the synthesized channel. Theorem \ref{theorem:Gaussian channel synthesis} implies $\mathbf{T}=\mathbf{T}_{2}\mathbf{T}_{1} = \mathbf{I}_{2}$, $\mathbf{d} = \mathbf{T}_{2}\mathbf{d}_{1}+\mathbf{d}_{2}=\mathbf{0}$ and 
\begin{align}
\mathbf{N} &= \mathbf{T}_{2}\mathbf{N}_{1}\mathbf{T}_{2}^{T}+\mathbf{N}_{2}
\nonumber\\ 
&= \frac{1}{\eta}(1-\eta)\Big{(}\nth+\frac{1}{2}\Big{)}\mathbf{I}_{2} + \frac{1}{2}\Big{(}\frac{1-\eta}{\eta}\Big{)}\mathbf{I}_{2} 
\nonumber\\
&=   \Big{(}\frac{1-\eta}{\eta}\Big{)}(\nth+1)\mathbf{I}_{2} . \label{eq:noise in loss then amp equals displacement}
\end{align}
Comparing Eq.\eqref{eq:noise in loss then amp equals displacement} with $\mathbf{N}=\sigma^{2}\mathbf{I}_{2}$ for the Gaussian random displacement channel, the statement follows. 
\end{proof}
It is noteworthy that the noise from the Gaussian loss channel (i.e., $\mathbf{N}_{1}=(1-\eta)(\nth +\frac{1}{2})\mathbf{I}_{2}$) is amplified by the quantum-limited amplification (first term in the second line of Eq.\eqref{eq:noise in loss then amp equals displacement}), hence increasing the variance $\sigma_{\eta,\nthtiny}^{2}$ of resulting random displacement channel. We can however avoid the noise amplification simply by reversing the order of loss channel and amplification, i.e., amplifying the signal prior to sending it through the Gaussian loss channel. 
\begin{lemma}[Pre-amplification causes less noise]
Let $\mathcal{N}[\eta,\nth]$ and $\mathcal{A}[1/\eta]$ be as specified in corollary \ref{corollary:loss plus amplification is displacement post-amplification}. Then, 
\begin{equation}
\mathcal{N}[\eta,\nth] \cdot \mathcal{A}[1/\eta] = \mathcal{N}_{B_{2}}[\tilde{\sigma}^{2}_{\eta,\nthtiny}], 
\end{equation}
where 
\begin{equation}
\tilde{\sigma}^{2}_{\eta,\nthtiny} \equiv (1-\eta)(\nth +1). \label{eq:effective variance of displacement pre-amp} 
\end{equation}
Note that $\tilde{\sigma}^{2}_{\eta,\nthtiny}$ is strictly less than $\sigma^{2}_{\eta,\nthtiny}$ for all $0 \le \eta<  1$.   \label{lemma:loss plus amplification is displacement pre-amplification}
\end{lemma} 
\begin{proof}
The proof goes similarly as in corollary \ref{corollary:loss plus amplification is displacement post-amplification}, except that the subscripts are exchanged $1\leftrightarrow 2$: 
\begin{align}
\mathbf{N} &= \mathbf{T}_{1}\mathbf{N}_{2}\mathbf{T}_{1}^{T}+\mathbf{N}_{1}
\nonumber\\ 
&=  \eta \times  \frac{1}{2}\Big{(}\frac{1-\eta}{\eta}\Big{)}\mathbf{I}_{2} + (1-\eta)\Big{(}\nth +\frac{1}{2}\Big{)}\mathbf{I}_{2} 
\nonumber\\
&=  (1-\eta)(\nth+1)\mathbf{I}_{2} . \label{eq:noise in loss then amp equals displacement pre-amplification}
\end{align}
\end{proof} 
Note that in Eq.\eqref{eq:effective variance of displacement pre-amp}, noise from the Gaussian loss channel (i.e., $\mathbf{N}_{1}$) is not amplified. Moreover, additional noise $\mathbf{N}_{2}$ from the quantum-limited amplification is now reduced by the factor of $\eta$ as compared with Eq.\eqref{eq:effective variance of displacement post-amp}, due to the loss (see the first term in the second line of Eq.\eqref{eq:noise in loss then amp equals displacement pre-amplification}). 

In subsection \ref{subsection:GKP code rate}, we combine Lemma \ref{lemma:loss plus amplification is displacement pre-amplification} with the earlier result given in \cite{Harrington2001} to establish the achievable quantum communication rate of the GKP codes for the general Gaussian loss channel $\mathcal{N}[\eta,\nth]$. 

Another interesting application of the idea of Gaussian channel synthesis is the composition of bosonic pure-loss channel and quantum-limited amplification with unmatched transmissivity and gain (i.e., $G\neq 1/\eta$). With this mismatch, one can simulate the general Gaussian loss channel $\mathcal{N}[\eta,\nth]$ with $\mathcal{N}[\eta',0]$ and $\mathcal{A}[G']$ for some properly chosen $\eta',G'$ \cite{Caruso2006B,Raul2012}:
\begin{lemma}[Eq.(5.1) of \cite{Sharma2017}] 
Gaussian loss channel can be decomposed into a bosonic pure-loss channel followed by a quantum-limited amplification 
\begin{equation}
\mathcal{N}[\eta,\nth] = \mathcal{A}[G']\cdot \mathcal{N}[\eta',0], 
\end{equation}
where $\eta' = \frac{\eta}{(1-\eta)\nthtiny + 1}$ and $G'=(1-\eta)\nth +1$.  \label{lemma:Gaussian loss channel decomposed into pure loss and amplficiation post-amplification}
\end{lemma} 

In \cite{Sharma2017}, Lemma \ref{lemma:Gaussian loss channel decomposed into pure loss and amplficiation post-amplification} was combined with the data-processing argument to upper bound the quantum capacity of Gaussian loss channel $\mathcal{N}[\eta,\nth]$. In subsection \ref{subsection:Improved upper bound of Gaussian loss channel capacity}, we give a tighter upper bound by introducing a slight modification of this approach, i.e., reversing the order of bosonic pure-loss channel and amplification (see Theorem \ref{theorem:improved upper bound of quantum capacity energy unconstrained} and Theorem \ref{theorem:improved upper bound of quantum capacity energy constrained}).  
\begin{lemma}
Reverse the order of bosonic pure-loss channel and quantum-limited amplification in Lemma \ref{lemma:Gaussian loss channel decomposed into pure loss and amplficiation post-amplification}. Then, 
\begin{equation}
\mathcal{N}[\eta,\nth] = \mathcal{N}[\tilde{\eta}',0]\cdot \mathcal{A}[\tilde{G}'], 
\end{equation}
where $\tilde{\eta}' = \eta- (1-\eta)\nth$ and $\tilde{G}' = \frac{\eta}{\eta-(1-\eta)\nthtiny}$.  \label{lemma:Gaussian loss channel decomposed into pure loss and amplficiation pre-amplification}
\end{lemma}
\begin{proof}
Recall that $\mathcal{A}[\tilde{G}']$ and $\mathcal{N}[\tilde{\eta}',0]$ are characterized by $(\mathbf{T}_{1},\mathbf{N}_{1},\mathbf{d}_{1}) = (\sqrt{\tilde{G}'}\mathbf{I}_{2},\frac{(\tilde{G}'-1)}{2}\mathbf{I}_{2},\mathbf{0})$ and $(\mathbf{T}_{2},\mathbf{N}_{2},\mathbf{d}_{2}) = (\sqrt{\tilde{\eta}'}\mathbf{I}_{2},\frac{1}{2}(1-\tilde{\eta}')\mathbf{I}_{2},\mathbf{0})$, respectively. Then, the synthesized channel is a Gaussian channel characterized by $\mathbf{T} = \mathbf{T}_{2}\mathbf{T}_{1} = \sqrt{\tilde{\eta}'\tilde{G}'}\mathbf{I}_{2} = \sqrt{\eta}\mathbf{I}_{2}$, $\mathbf{d}=\mathbf{T}_{2}\mathbf{d}_{1}+\mathbf{d}_{2}=\mathbf{0}$ and 
\begin{align}
\mathbf{N} &= \mathbf{T}_{2}\mathbf{N}_{1}\mathbf{T}_{2}^{T}+\mathbf{N}_{2}
\nonumber\\ 
&=  \tilde{\eta}'\frac{(\tilde{G}'-1)}{2}\mathbf{I}_{2} +   \frac{1}{2}(1-\tilde{\eta}')\mathbf{I}_{2} 
\nonumber\\
&=  \frac{1}{2}(1+\eta)\mathbf{I}_{2}  - \tilde{\eta}'\mathbf{I}_{2} =  (1-\eta)\Big{(}\nth+\frac{1}{2}\Big{)}\mathbf{I}_{2},  
\end{align}
i.e., identical specification to that of $\mathcal{N}[\eta,\nth]$. 
\end{proof}
We note that Lemma \ref{lemma:Gaussian loss channel decomposed into pure loss and amplficiation pre-amplification} was independently discovered and stated in Theorem 30 of \cite{Sharma2017} and in Lemma 1 of \cite{Rosati2018}. 

\section{Quantum capacity of Gaussian channels}
\label{section:Quantum capacity of Gaussian channels}

\subsection{Earlier results on Gaussian quantum channel capacity}

Let $\mathcal{N} : \mathcal{L}(\mathcal{H})\rightarrow \mathcal{L}(\mathcal{H})$ be a noisy quantum channel from an information sender to a receiver, dilated by $\mathcal{N}(\hat{\rho}) \equiv \textrm{Tr}_{E}[\hat{U}(\hat{\rho}\otimes \hat{\rho}_{E})\hat{U}^{\dagger}]$, where $\hat{U}$ is a unitary operator acting on $\mathcal{H}\otimes\mathcal{H}_{E}$, and $\mathcal{H}_{E}$ is the Hilbert space of the environment, causing the channel noise. Complementary channel $\mathcal{N}^{c}:\mathcal{L}(\mathcal{H})\rightarrow \mathcal{L}(\mathcal{H}_{E})$ is the channel from the sender to the environment, i.e., $\mathcal{N}^{c}(\hat{\rho}) \equiv \textrm{Tr}_{S}[\hat{U}(\hat{\rho}\otimes \hat{\rho}_{E})\hat{U}^{\dagger}]$. Coherent information of a channel is defined as $Q(\mathcal{N}) \equiv \max_{\hat{\rho}\in\mathcal{D}(\mathcal{H})} I_{c}(\hat{\rho},\mathcal{N})$, where $I_{c}(\hat{\rho},\mathcal{N}) \equiv  H(\mathcal{N}(\hat{\rho})) - H(\mathcal{N}^{c}(\hat{\rho}))$ and $H(\hat{\rho})\equiv -\textrm{Tr}[\hat{\rho}\log\hat{\rho}]$ is the entropy of a state $\hat{\rho}$, where $\log$ is the logarithm with base $2$. Quantum capacity of a channel $\mathcal{N}$ equals to the regularized coherent information of the channel: $Q_{\reg}(\mathcal{N}) \equiv  \lim_{n\rightarrow \infty}\frac{1}{n}Q(\mathcal{N}^{\otimes n})$ \cite{Schumacher1996,Lloyd1997,Devetak2005}. Coherent information might be superadditive, i.e., $Q(\mathcal{N}\otimes \mathcal{N}') \ge Q(\mathcal{N}) + Q(\mathcal{N}')$, and thus the one-shot expression $Q(\mathcal{N})$ only lower bounds the true quantum capacity: $Q(\mathcal{N})\le Q_{\reg}(\mathcal{N})$. 

A quantum channel $\mathcal{N}$ is degradable iff there exists a degrading channel $\mathcal{D}:\mathcal{L}(\mathcal{H})\rightarrow \mathcal{L}(\mathcal{H}_{E})$ such that $\mathcal{N}^{c} = \mathcal{D}\cdot \mathcal{N}$, i.e., iff the receiver can simulate complementary channel. Coherent information of degradable channels are additive (see, e.g., Theorem 12.5.4 in \cite{Wilde2013}), and thus the one-shot coherent information fully characterizes their quantum capacity: $Q_{\reg}(\mathcal{N})=Q(\mathcal{N})$ if $\mathcal{N}$ is degradable. A quantum channel $\mathcal{N}$ is anti-degradable iff there exists a degrading channel $\mathcal{D}:\mathcal{L}(\mathcal{H}_{E})\rightarrow \mathcal{L}(\mathcal{H})$ such that $\mathcal{N} = \mathcal{D}\cdot \mathcal{N}^{c}$, i.e., iff the environment can simulate the channel from sender to receiver. Anti-degradable channels have zero quantum capacity and thus are unable to support reliable information transmission.    

Complementary channel of the bosonic pure-loss channel with transmissivity $\eta$ is also a bosonic pure-loss channel: $\mathcal{N}^{c}[\eta,0] = \mathcal{N}[1-\eta,0]$. Since bosonic pure-loss channels are multiplicative $\mathcal{N}[\eta_{1},0]\cdot\mathcal{N}[\eta_{2},0]=\mathcal{N}[\eta_{1}\eta_{2},0]$, as can be justified by applying Theorem \ref{theorem:Gaussian channel synthesis}, they are degradable if $\eta \in [\frac{1}{2},1]$, i.e., $\mathcal{N}^{c}[\eta,0] = \mathcal{N}[1-\eta,0] = \mathcal{N}[\frac{1-\eta}{\eta},0]\cdot \mathcal{N}[\eta,0]$ (similarly, anti-degradable if $\eta\in[0,\frac{1}{2}]$). 

Evaluation of the coherent information of a channel involves optimization over all input state $\hat{\rho}$. For degradable Gaussian channels, it was proven that optimization over all Gaussian states is sufficient \cite{Wolf2007}. For Gaussian loss channels, optimization of one-shot coherent information over thermal states was performed in \cite{Holevo2001}. As was pointed out in Remark 17 of \cite{Wilde2016}, however, sufficiency of thermal optimizer is yet to be justified. Here, we justify this for bosonic pure-loss channels by using rotational invariance of bosonic pure-loss channels and concavity of coherent information for degradable channels. 

Define the Gaussian unitary rotation channel $\mathcal{U}[\theta]$ as $\mathcal{U}[\theta](\hat{\rho}) \equiv \hat{U}(\theta)\hat{\rho}\hat{U}^{\dagger}(\theta)$, where $\hat{U}(\theta)\equiv e^{i\theta \hat{n}}$ is the phase rotation operator. Bosonic pure-loss channels are invariant under the phase rotation: $\mathcal{U}[\theta]\mathcal{N}[\eta,0] = \mathcal{N}[\eta,0]\mathcal{U}[\theta]$, as can be verified by specializing Theorem \ref{theorem:Gaussian channel synthesis} to $\mathcal{N}[\eta,0]$ and $\mathcal{U}[\theta]$, where the specification of the latter channel is given by $(\mathbf{T},\mathbf{N},\mathbf{d}) = (\mathbf{R}(\theta),\mathbf{0},\mathbf{0})$ and $\mathbf{R}(\theta)$ is given in Eq.\eqref{eq:symplectic transformation phase rotation}. We then prove the following theorem. 
\begin{theorem}[Sufficiency of thermal optimizer for bosonic pure-loss channel capacity]
Let $\mathcal{N}[\eta,0]$ be a bosonic pure-loss channel with transimissivity $\eta\in[0,1]$ (see Definition \ref{definition:Gaussian loss channels}) and $\hat{\rho}=\sum_{m,n=0}^{\infty}\rho_{mn}|m\rangle\langle n|$ be an arbitrary bosonic state represented in the Fock basis. Then, 
\begin{equation}
I_{c}(\hat{\rho},\mathcal{N}[\eta,0]) \le I_{c}(\sum_{n=0}^{\infty}\rho_{nn}|n\rangle\langle n|,\mathcal{N}[\eta,0]). \label{eq:thermal optimizer concavity rotation twirling}
\end{equation}
Diagonal states in the Fock basis are thus sufficient for the maximization of coherent information of bosonic pure-loss channels. Combining this with the sufficiency of Gaussian input states \cite{Wolf2007}, it follows that coherent information of bosonic pure-loss channel is maximized by thermal states. \label{theorem:thermal optimizer is enough for bosonic pure loss channel capacity}
\end{theorem} 
\begin{proof}
Define $\hat{\rho}_{\theta} \equiv \mathcal{U}[\theta]\hat{\rho} = e^{i\theta\hat{n}}\hat{\rho}e^{-i\theta\hat{n}}$ and let $p(\theta)$ be a probability density defined over $\theta\in [0,2\pi)$. Since bosonic pure-loss channel is degradable, 
its coherent information $I_{c}(\hat{\rho},\mathcal{N}[\eta,0])$ is concave in the input state (see Theorem 12.5.6 in \cite{Wilde2013}): 
\begin{equation}
\int_{0}^{2\pi} d\theta p(\theta)I_{c}(\hat{\rho}_{\theta},\mathcal{N}[\eta,0]) \le I_{c}\Big{(}\int_{0}^{2\pi} d\theta p(\theta)\hat{\rho}_{\theta},\mathcal{N}[\eta,0]\Big{)},  \label{eq:thermal sufficiency proof intermediate}
\end{equation}
Rotational invariance of Gaussian loss channel implies $\mathcal{N}[\eta,0](\hat{\rho}_{\theta}) = \mathcal{U}[\theta]\mathcal{N}[\eta,0](\hat{\rho})$ and similarly $\mathcal{N}^{c}[\eta,0](\hat{\rho}_{\theta}) = \mathcal{U}[\theta]\mathcal{N}^{c}[\eta,0](\hat{\rho})$. Since the quantum entropy is invariant under unitary operation, we have $I_{c}(\hat{\rho}_{\theta},\mathcal{N}[\eta,0]) = I_{c}(\hat{\rho},\mathcal{N}[\eta,0])$. The left hand side of Eq.\eqref{eq:thermal sufficiency proof intermediate} is then given by $I_{c}(\hat{\rho},\mathcal{N}[\eta,0])$ since $\int_{0}^{2\pi}p(\theta)=1$. Choosing $p(\theta)$ to be a flat distribution $p(\theta) = 1/(2\pi)$, we find 
\begin{align}
\int_{0}^{2\pi} d\theta p(\theta) \hat{\rho}_{\theta} &= \sum_{m,n=0}^{\infty}\frac{1}{2\pi} \int_{0}^{2\pi}  d\theta e^{i\theta(m-n)} \rho_{mn}|m\rangle\langle n|
\nonumber\\
&= \sum_{n=0}^{\infty} \rho_{nn}|n\rangle\langle n|, \label{eq:rotation twirling} 
\end{align}
where we used $\int_{0}^{2\pi}  d\theta e^{i\theta(m-n)}=2\pi \delta_{mn}$ to derive the last equality. Plugging Eq.\eqref{eq:rotation twirling} into the right hand side of Eq.\eqref{eq:thermal sufficiency proof intermediate}, Eq.\eqref{eq:thermal optimizer concavity rotation twirling} follows. 

We emphasize that the phase rotation $\mathcal{U}[\theta]$ does not change the average photon number of a state. Thus, the above argument also applies to an energy-constrained scenario, where average photon number is bounded from above by $\bar{n}$.   
\end{proof}

Note that the thermal state $\hat{\rho}_{\bar{n}}$ with average photon number $\bar{n}$ is transformed by bosonic pure-loss channel into another thermal state with $\bar{n}'=\eta \bar{n}$: $\mathcal{N}[\eta,0](\hat{\rho}_{\bar{n}}) = \hat{\rho}_{\eta\bar{n}}$. Similarly, $\mathcal{N}^{c}[\eta,0](\hat{\rho}_{\bar{n}}) = \hat{\rho}_{(1-\eta)\bar{n}}$. Since $H(\hat{\rho}_{\bar{n}}) = g(\bar{n})$, where 
\begin{equation}
g(x)\equiv (x+1)\log(x+1)-x\log x
\end{equation}
(see Eq.\eqref{eq:thermal state entropy} and the text below in Appendix \ref{section:Bosonic mode, Gaussian states and Gaussian unitary operations}), Theorem \ref{theorem:thermal optimizer is enough for bosonic pure loss channel capacity} leads to the following energy-constrained quantum capacity of bosonic pure-loss channel: (Eq.(1) in \cite{Wilde2016})
\begin{align}
Q^{n \le \bar{n}}_{\reg}(\mathcal{N}[\eta,0]) &= Q^{n \le \bar{n}}(\mathcal{N}[\eta,0]) 
\nonumber\\
&= \max[ g(\eta \bar{n})-g((1-\eta)\bar{n}) , 0],  \label{eq:quantum capacity bosonic pure loss channel enenrgy constrained}
\end{align}
where $\bar{n}$ is the maximum allowed photon number per bosonic mode and should not to be confused with $\nth$ in $\mathcal{N}[\eta,\nth]$. When deriving the last equality of Eq.\eqref{eq:quantum capacity bosonic pure loss channel enenrgy constrained}, we used the fact that $g(\eta \bar{n})-g((1-\eta)\bar{n})$ increases monotonically in $\bar{n}$ (see Remark 21 of \cite{Sharma2017}) and thus the optimal input state (with average photon number less than $\bar{n}$) which maximizes the coherent information is given by $\hat{\rho}=\hat{\rho}_{\bar{n}}$. In the energy-unconstrained case $\bar{n}\rightarrow \infty$, Eq.\eqref{eq:quantum capacity bosonic pure loss channel enenrgy constrained} reduces to 
\begin{equation}
Q_{\reg}(\mathcal{N}[\eta,0]) = Q(\mathcal{N}[\eta,0])  = \max\Big{[} \log\Big{(}\frac{\eta}{1-\eta}\Big{)}, 0\Big{]}.   \label{eq:quantum capacity bosonic pure loss channel enenrgy un-constrained}
\end{equation} 

The general Gaussian loss channel $\mathcal{N}[\eta,\nth]$ (see Definition \ref{definition:Gaussian loss channels}) with $\nth\neq 0$ is not degradable nor anti-degradable \cite{Caruso2006B},\cite{Holevo2007}, and thus its coherent information may not be additive. A lower bound of quantum capacity of Gaussian loss channel can be obtained by evaluating the one-shot coherent information with a thermal input state \cite{Holevo2001}. In the energy-constrained case $(n\le\bar{n})$,
\begin{align}
I_{c}(\hat{\rho}_{\bar{n}},\mathcal{N}[\eta,\nth]) &=  g(\eta \bar{n}+(1-\eta)\nth)
\nonumber\\
&- g\Big{(}\frac{D+(1-\eta)(\bar{n}-\nth) -1}{2}\Big{)} 
\nonumber\\
&- g\Big{(}\frac{D-(1-\eta)(\bar{n}-\nth) -1}{2}\Big{)} , 
\end{align}
where $D\equiv \sqrt{((1+\eta)\bar{n}+(1-\eta)\nth+1)^{2}-4\eta\bar{n}(\bar{n}+1)}$. In $\bar{n}\rightarrow\infty$ limit, this reduces to 
\begin{equation}
\lim_{\bar{n}\rightarrow\infty } I_{c}(\hat{\rho}_{\bar{n}},\mathcal{N}[\eta,\nth]) = \log\Big{(}\frac{\eta}{1-\eta}\Big{)} -g(\nth).  
\end{equation}
The best known lower bound of the quantum capacity of Gaussian loss channel is thus   
\begin{align}
Q^{n\le \bar{n}}_{\reg}(\mathcal{N}[\eta,\nth]) &\ge  \max [ I_{c}(\hat{\rho}_{\bar{n}},\mathcal{N}[\eta,\nth]) , 0],
\nonumber\\
Q_{\reg}(\mathcal{N}[\eta,\nth]) &\ge  \max \Big{[} \lim_{\bar{n}\rightarrow\infty} I_{c}(\hat{\rho}_{\bar{n}},\mathcal{N}[\eta,\nth]) , 0\Big{]}, \label{eq:quantum capacity Gaussian loss channel lower bound}
\end{align}
in the energy-constrained and unconstrained case, respectively. 

We note that Gaussian random displacement channel can be understood as a certain limit of Gaussian loss channel $\mathcal{N}_{B_{2}}[\sigma^{2}] = \lim_{\eta\rightarrow 1} \mathcal{N}[\eta,\frac{\sigma^{2}}{1-\eta}-\frac{1}{2}]$. Thus, lower bound of its quantum capacity is given by 
\begin{equation}
Q_{\reg}(\mathcal{N}_{B_{2}}[\sigma^{2}]) \ge \max \Big{[} \log\Big{(} \frac{1}{e\sigma^{2}} \Big{)} ,0 \Big{]} \label{eq:quantum capacity Gaussian displacement channel lower bound}
\end{equation}
in the energy-unconstrained case. 

General upper bound of quantum capacity was introduced in \cite{Holevo2001} and applied to Gaussian loss channels (Eq.(5.7) in \cite{Holevo2001}, translated into our notation): 
\begin{align}
Q_{\reg}(\mathcal{N}[\eta,\nth])  &\le  Q_{\scriptsize{\textrm{HW}}}(\eta,\nth)
\nonumber\\
&\equiv \max\Big{[} \log\Big{(} \frac{1+\eta}{(1-\eta)(2\nth + 1)} \Big{)} ,0\Big{]}. \label{eq:Holevo upper bound} 
\end{align}
However, this upper bound is not tight: In the case of bosonic pure-loss channel, $Q_{\scriptsize{\textrm{HW}}}(\eta,0) = \log ( \frac{1+\eta}{1-\eta}) > Q(\mathcal{N}[\eta,0])$ for all $\eta \in [0,1)$, where $Q(\mathcal{N}[\eta,0])$ is given in Eq.\eqref{eq:quantum capacity bosonic pure loss channel enenrgy un-constrained}. Recently, three new upper bounds were obtained in \cite{Sharma2017}, where one of them is based on Lemma \ref{lemma:Gaussian loss channel decomposed into pure loss and amplficiation post-amplification} and data-processing argument, and the other two are based on the notion of approximate degradability of quantum channels developed in \cite{Sutter2017}. Here, we only present the data-processing bound, in the energy-constrained form, while referring to the original paper for the other two approximate degradability bounds: (Theorems 19 and 24 in \cite{Sharma2017})
\begin{equation}
Q^{n\le \bar{n}}_{\reg}(\mathcal{N}[\eta,\nth]) \le Q^{n\le \bar{n}}_{\scriptsize{\textrm{DP}}}(\eta,\nth) \equiv Q^{n\le \bar{n}}(\mathcal{N}[\eta',0]), \label{eq:quantum capacity of Gaussian loss channel data processing upper bound post amplification}
\end{equation}
where $Q^{n\le \bar{n}}_{\scriptsize{\textrm{DP}}}(\eta,\nth)$ is the data-processing bound, $\eta' = \frac{\eta}{(1-\eta)\nthtiny+1}$ and $Q^{n\le \bar{n}}(\mathcal{N}[\eta',0])$ is given in Eq.\eqref{eq:quantum capacity bosonic pure loss channel enenrgy constrained}.
\begin{proof}
Since the regularized coherent information is an achievable quantum communication rate, there exists a set of encoding and decoding channels, denoted by $\lbrace \mathcal{E},\mathcal{D}\rbrace$, which achieves the communication rate $R = Q^{n\le \bar{n}}_{\reg}(\mathcal{N}[\eta,\nth])$ for Gaussian loss channel $\mathcal{N}[\eta,\nth]$. Since $\mathcal{N}[\eta,\nth] = \mathcal{A}[G']\cdot\mathcal{N}[\eta',0]$ (see Lemma \ref{lemma:Gaussian loss channel decomposed into pure loss and amplficiation post-amplification}), this implies that the encoding and decoding set $\lbrace \mathcal{E} , \mathcal{D}\cdot\mathcal{A}[G'] \rbrace$ achieves the rate $R = Q^{n\le \bar{n}}_{\reg}(\mathcal{N}[\eta,\nth])$ for the bosonic pure-loss channel $\mathcal{N}[\eta',0]$. Since the achievable rate $R$ is upper bounded by the quantum capacity $Q^{n\le \bar{n}}(\mathcal{N}[\eta',0])$, Eq.\eqref{eq:quantum capacity of Gaussian loss channel data processing upper bound post amplification} follows. 
\end{proof}
Note that $Q^{n\le \bar{n}}_{\scriptsize{\textrm{DP}}}(\eta,\nth)$ converges to $Q^{n\le \bar{n}}(\mathcal{N}[\eta,0])$ in $\nth\rightarrow 0$ limit, as $\eta'=\eta$ when $\nth=0$. Similarly, an upper bound of the quantum capacity of Gaussian random displacement channel $\mathcal{N}_{B_{2}}[\sigma^{2}]$ can be obtained by combining the data-processing argument and Eq.\eqref{eq:pure loss plus amplification is displacement}: 
\begin{equation}
Q_{\reg}(\mathcal{N}_{B_{2}}[\sigma^{2}]) \le \max \Big{[} \log\Big{(} \frac{1}{\sigma^{2}} \Big{)} ,0 \Big{]}. \label{eq:quantum capacity upper bound displacement channel less tight}
\end{equation}
Alternatively, this bound can be derived from Eq.\eqref{eq:Holevo upper bound} by taking the limit $\eta\rightarrow 1$ and $\nth \rightarrow \frac{\sigma^{2}}{1-\eta}-\frac{1}{2}$. 

\subsection{Improved upper bound of Gaussian loss channel capacity}
\label{subsection:Improved upper bound of Gaussian loss channel capacity}

Here, we improve the data-processing bound slightly by using Lemma \ref{lemma:Gaussian loss channel decomposed into pure loss and amplficiation pre-amplification}, instead of Lemma \ref{lemma:Gaussian loss channel decomposed into pure loss and amplficiation post-amplification}, for the decomposition of Gaussian loss channel.  

\begin{theorem}[Improved data-processing bound]
In the energy-unconstrained case, quantum capacity of Gaussian loss channel $\mathcal{N}[\eta,\nth]$ (see Definition \ref{definition:Gaussian loss channels}) is upper bounded by the improved data-processing bound $Q_{\scriptsize{\textrm{IDP}}}(\eta,\nth)$: 
 
\begin{equation}
Q_{\reg}(\mathcal{N}[\eta,\nth]) \le  Q_{\scriptsize{\textrm{IDP}}}(\eta,\nth) \equiv   Q_{\reg}(\mathcal{N}[\tilde{\eta}',0]),  \label{eq:improved upper bound energy unconstrained less simplified} 
\end{equation}
where $\tilde{\eta}' = \eta-(1-\eta)\nth$ and $Q_{\reg}(\mathcal{N}[\tilde{\eta}',0])$ is given in Eq.\eqref{eq:quantum capacity bosonic pure loss channel enenrgy un-constrained}.  improved data-processing bound $Q_{\scriptsize{\textrm{IDP}}}(\eta,\nth) $ simplifies to 
\begin{equation}
Q_{\scriptsize{\textrm{IDP}}}(\eta,\nth)  = \max\Big{[} \log \Big{(} \frac{\eta-(1-\eta)\nth}{(1-\eta)(\nth+1)} \Big{)},0 \Big{]} . \label{eq:improved upper bound of quantum capacity energy unconstrained simplified}
\end{equation}
\label{theorem:improved upper bound of quantum capacity energy unconstrained}
\end{theorem}
\begin{proof}
The proof goes in the same way as above, except that Lemma \ref{lemma:Gaussian loss channel decomposed into pure loss and amplficiation post-amplification} is replaced by Lemma \ref{lemma:Gaussian loss channel decomposed into pure loss and amplficiation pre-amplification}: $\mathcal{N}[\eta,\nth] = \mathcal{N}[\tilde{\eta}',0]\cdot \mathcal{A}[\tilde{G}']$ with $\tilde{G}' = \eta/(\eta-(1-\eta)\nth)=\eta/\tilde{\eta}'$. Then, the encoding and decoding set $\lbrace  \mathcal{A}[\tilde{G}'] \cdot \mathcal{E} , \mathcal{D} \rbrace$ achieves the rate $R = Q_{\reg}(\mathcal{N}[\eta,\nth])$ for the bosonic pure-loss channel $\mathcal{N}[\tilde{\eta}',0]$. Since $R\le Q(\mathcal{N}[\tilde{\eta}',0])$, the theorem follows.   
\end{proof}

Note that 
\begin{equation}
Q_{\scriptsize{\textrm{IDP}}}(\eta,\nth) < Q_{\scriptsize{\textrm{DP}}}(\eta,\nth)\equiv\lim_{\bar{n}\rightarrow\infty} Q^{n\le \bar{n}}_{\scriptsize{\textrm{DP}}}(\eta,\nth) \label{eq:IDP better than DP in energy-unconstrained}
\end{equation}
for all $\eta\in[0,1)$, since 
\begin{align}
\eta' &= \frac{\eta}{(1-\eta)\nthtiny+1}
\nonumber\\
&= \eta - \frac{\eta (1-\eta)\nthtiny}{(1-\eta)\nthtiny+1} > \eta - (1-\eta)\nth = \tilde{\eta}'. \label{eq:comparison of etas in data-processing bound}
\end{align}
Thus, $Q_{\scriptsize{\textrm{IDP}}}(\eta,\nth)$ is a strictly tighter upper bound of Gaussian loss channel capacity than $Q_{\scriptsize{\textrm{DP}}}(\eta,\nth)$. We note that Theorem \ref{theorem:improved upper bound of quantum capacity energy unconstrained} was independently discovered in \cite{Rosati2018} (see Eqs.(39),(40) therein).

\begin{figure*}[!t]
\centering
\subfloat{\includegraphics[width=3.2in]{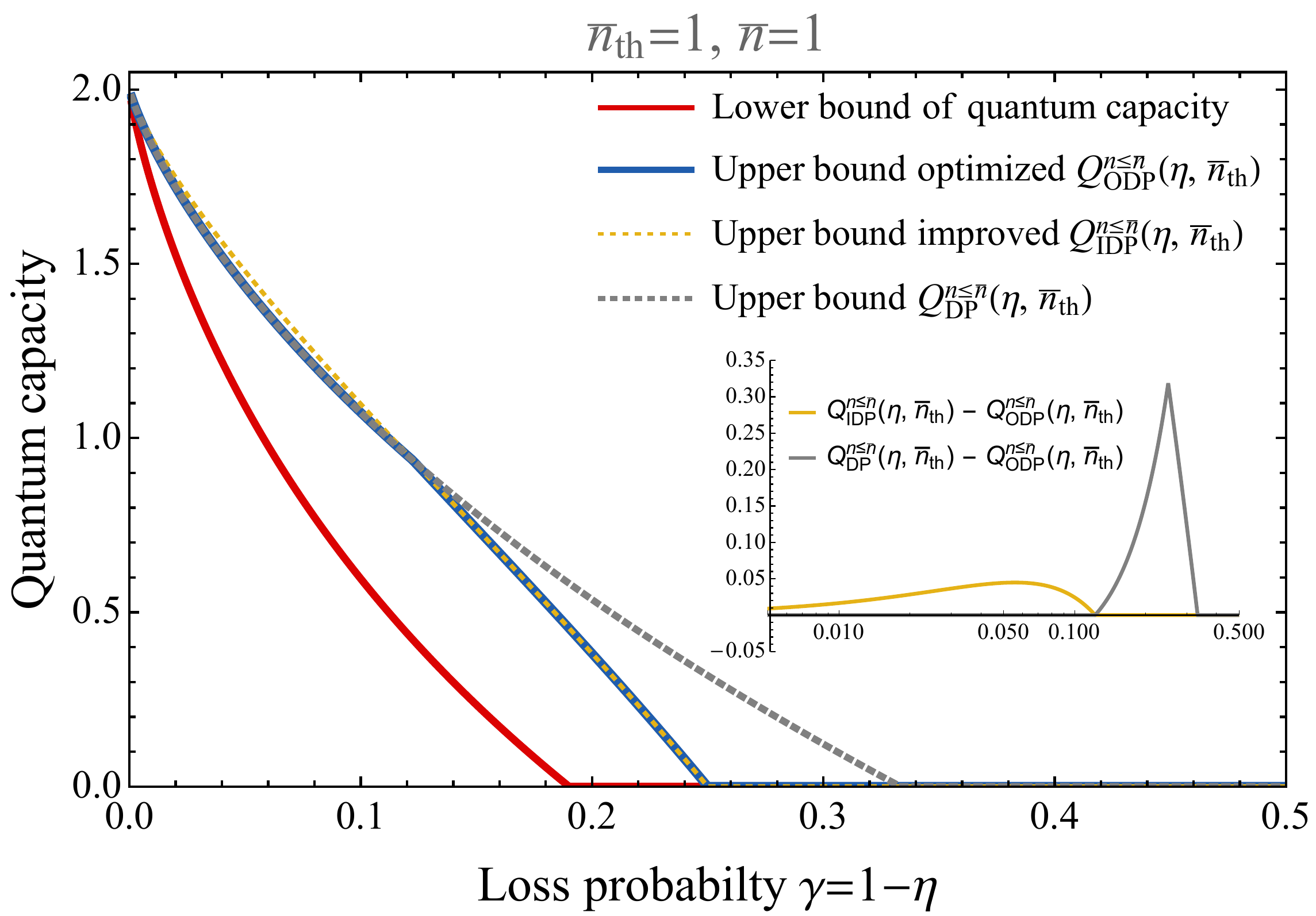}%
\label{fig:Quantum capacity bound combined nbar 01}}
\hfil
\subfloat{\includegraphics[width=3.2in]{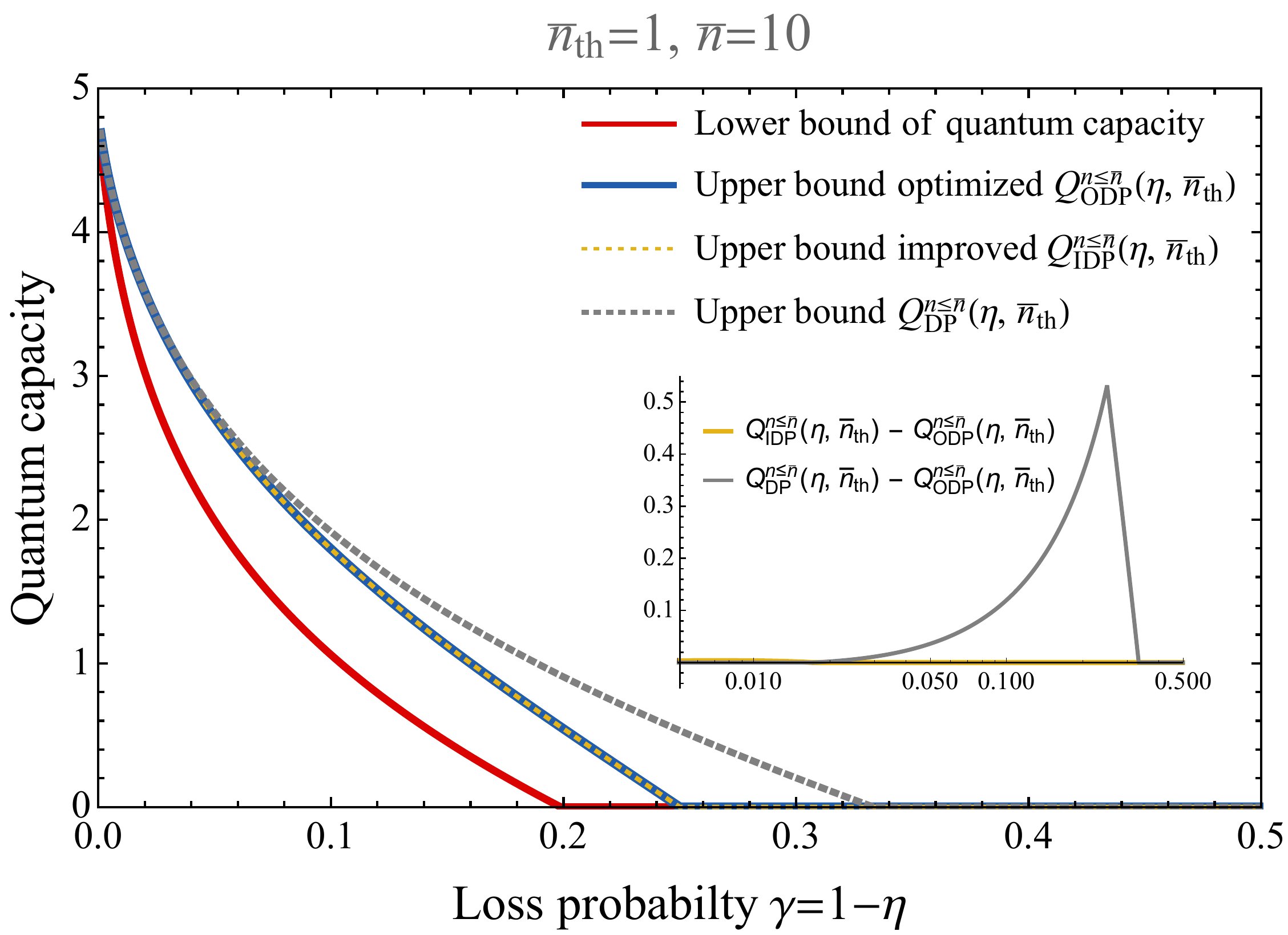}%
\label{fig:Quantum capacity bound combined nbar 10}}
\caption{Bounds of quantum capacity of Gaussian loss channel $\mathcal{N}[\eta,\nthtiny]$ for $(\nth,\bar{n})=(1,1)$ (left) and $(\nth,\bar{n})=(1,10)$ (right). Lower bound (red) is obtained by evaluating the coherent information with a thermal input state (see Eq.\eqref{eq:quantum capacity Gaussian loss channel lower bound}). The improved data-processing bound $Q^{n\le \bar{n}}_{\tiny{\textrm{IDP}}}(\eta,\nthtiny)$ (dashed yellow) is identical to the optimized data-processing bound $Q^{n\le \bar{n}}_{\tiny{\textrm{ODP}}}(\eta,\nthtiny)$ (blue) in a wide range of parameter space, and is very close to the optimal one even when it is not optimal. The data-processing bound $Q^{n\le \bar{n}}_{\tiny{\textrm{DP}}}(\eta,\nthtiny)$ (dashed grey) is optimal when $\eta \ge  \eta^{\star}(\nth,\bar{n})$ for some $\eta^{\star}(\nth,\bar{n})$ (e.g., $\eta^{\star}(1,1) = 0.8775\cdots$).   }
\label{fig:Quantum capacity bound combined}
\end{figure*}

In the energy-constrained case, more care is needed when combining Lemma \ref{lemma:Gaussian loss channel decomposed into pure loss and amplficiation pre-amplification} with the data-processing argument, because the quantum-limited amplification in $\mathcal{E}' = \mathcal{A}[\tilde{G}']\cdot\mathcal{E}$ increases the energy of the encoded state: 
\begin{equation}
\bar{n}' = \textrm{Tr}[\mathcal{A}[\tilde{G}'](\hat{\rho})\hat{n}] = \tilde{G}'\bar{n} +(\tilde{G}'-1), \label{eq:new photon number after amplification}
\end{equation}
as can be derived from the symplectic transformation of the two-mode squeezing (Eq.\eqref{eq:symplectic transformation two mode squeezer} and the text above in Appendix \ref{section:Bosonic mode, Gaussian states and Gaussian unitary operations}).

\begin{theorem}[Improved data-processing bound with energy constraint]
Energy-constrained quantum capacity of Gaussian loss channel $\mathcal{N}[\eta,\nth]$ is upper bounded by
\begin{align}
Q^{n\le \bar{n}}_{\reg}(\mathcal{N}[\eta,\nth]) &\le Q^{n\le \bar{n}}_{\scriptsize{\textrm{IDP}}}(\eta,\nth) 
\nonumber\\
&\equiv Q^{n\le \tilde{G}'\bar{n}+(\tilde{G}'-1)}(\mathcal{N}[\tilde{\eta}',0]), 
\end{align}
where $\bar{n}$ is the maximum allowed average photon number per bosonic mode in the encoded state. $Q^{n\le \bar{n}}_{\scriptsize{\textrm{IDP}}}(\eta,\nth)$ simplifies to 
\begin{align}
Q^{n\le \bar{n}}_{\scriptsize{\textrm{IDP}}}(\eta,\nth) &= \max\Big{[}g(\eta \bar{n}+(1-\eta)\nth) 
\nonumber\\
&- g\Big{(} \frac{ (1-\eta)(\nthtiny+1)(\eta\bar{n}+(1-\eta)\nthtiny)   }{\eta - (1-\eta)\nthtiny} \Big{)} , 0\Big{]}. 
\end{align}    \label{theorem:improved upper bound of quantum capacity energy constrained}
\end{theorem}
\begin{proof}
Let $\lbrace \mathcal{E}^{n\le \bar{n}},\mathcal{D} \rbrace$ be the set of encoding and decoding which achieves the rate $R=Q^{n\le \bar{n}}_{\reg}(\mathcal{N}[\eta,\nth])$ for the Gaussian loss channel $\mathcal{N}[\eta,\nth]$. Then, the encoding and decoding set $\lbrace \mathcal{A}[\tilde{G}']\cdot\mathcal{E}^{n\le \bar{n}} , \mathcal{D} \rbrace$ achieves the rate $R = Q^{n\le \bar{n}}_{\reg}(\mathcal{N}[\eta,\nth])$ for the bosonic pure-loss channel $\mathcal{N}[\tilde{\eta}',0]$. Since the new encoding $\mathcal{E}' \equiv \mathcal{A}[\tilde{G}']\cdot\mathcal{E}^{n\le \bar{n}}$ has photon number $n\le \tilde{G}'\bar{n}+(\tilde{G}'-1)$, the rate $R$ should be less than $Q^{n\le \tilde{G}'\bar{n}+(\tilde{G'}-1)}(\mathcal{N}[\tilde{\eta}',0])$.  
\end{proof}

We note that Theorem \ref{theorem:improved upper bound of quantum capacity energy constrained} was independently discovered in \cite{Sharma2017} (see Theorem 45 therein). We emphasize that, unlike the energy-unconstrained case, our new upper bound $Q^{n\le \bar{n}}_{\scriptsize{\textrm{IDP}}}(\eta,\nth)$ is not always tighter than $Q^{n\le \bar{n}}_{\scriptsize{\textrm{DP}}}(\eta,\nth)$ (see the inset of Fig. \ref{fig:Quantum capacity bound combined}). Physically, this is because the increased encoding energy due to pre-amplification allows larger quantum capacity, which is crucial if the allowed average photon number in the encoding is small, i.e., $\bar{n}\ll 1$. To overcome this issue, we consider a more general decomposition of the Gaussian loss channel, including both pre-amplification and post-amplification:  
\begin{equation}
\mathcal{N}[\eta,\nth] =\mathcal{A}[G_{1}] \mathcal{N}[\bar{\eta},0]\mathcal{A}[G_{2}], \label{eq:decomposition of thermal loss into preamp pure loss and post amp}
\end{equation}
where $\bar{\eta} = 1-\frac{(1-\eta)}{G_{1}}(\nth+1)$, $G_{2} = \eta/ ( G_{1}-(1-\eta)(\nth+1) )$ and $G_{1}$ can take any value in the range $1\le G_{1} \le 1+(1-\eta)\nth$ (the upper bound of $G_{1}$ is imposed by $G_{2}\ge 1$). In this more general setting, the upper bound of the quantum capacity of $\mathcal{N}[\eta,\nth]$ is given by $Q^{n\le G_{2}\bar{n}+(G_{2}-1)}(\mathcal{N}[\bar{\eta},0])$. This upper bound can then be optimized by choosing $G_{1},G_{2}\ge 1$ such that the upper bound is minimized, i.e., 
\begin{equation}
Q^{n\le \bar{n}}_{\reg}(\mathcal{N}[\eta,\nth]) \le \min_{G_{1},G_{2}\ge 1} Q^{n\le G_{2}\bar{n}+(G_{2}-1)}(\mathcal{N}[\bar{\eta},0]). 
\end{equation} 
\begin{theorem}[Optimized data-processing bound with energy-constraint]
Energy-constrained quantum capacity of Gaussian loss channel $\mathcal{N}[\eta,\nth]$ is upper bounded by the optimized data-processing bound $Q^{n\le \bar{n}}_{\scriptsize{\textrm{ODP}}}(\eta,\nth)$:
\begin{align}
Q^{n\le \bar{n}}_{\reg}(\mathcal{N}[\eta,\nth])  &\le Q^{n\le \bar{n}}_{\scriptsize{\textrm{ODP}}}(\eta,\nth)
\nonumber\\
&\equiv \min_{1 \le G_{1} \le 1+(1-\eta)\nth} f_{\eta,\nth,\bar{n}}(G_{1}), \label{eq:optimized data-processing bound}   
\end{align}
where $\bar{n}$ is the maximum allowed average photon number in the encoding and $f_{\eta,\nth,\bar{n}}(G_{1})$ is defined as 
\begin{align}
f_{\eta,\nth,\bar{n}}(G_{1}) &\equiv \max\Big{[}  g\big{(}\bar{\eta}(G_{2}\bar{n}+(G_{2}-1))\big{)}
\nonumber\\
&\qquad - g\big{(}(1-\bar{\eta})(G_{2}\bar{n}+(G_{2}-1))\big{)},0 \Big{]} . \label{eq:definition of the f function}  
\end{align} 
and $\bar{\eta} = 1-\frac{(1-\eta)}{G_{1}}(\nth+1)$, $G_{2} = \eta/ ( G_{1}-(1-\eta)(\nth+1) )$.   \label{theorem:optimized data-processing upper bound}
\end{theorem}
\begin{proof}
See appendix \ref{section:Derivation of the optimized data-processing bound} for the derivation. 
\end{proof}  

Since $Q^{n\le \bar{n}}_{\scriptsize{\textrm{DP}}}(\eta,\nth)$ and $Q^{n\le \bar{n}}_{\scriptsize{\textrm{IDP}}}(\eta,\nth)$ are the two extremes with $G_{2}=1$ and $G_{1}=1$, respectively, they are greater than or equal to the optimized data-processing bound. Thus $Q^{n\le \bar{n}}_{\scriptsize{\textrm{ODP}}}(\eta,\nth)$ is the best upper bound among all the data-processing type bounds introduced above (see Fig. \ref{fig:Quantum capacity bound combined}). We numerically observe, however, that either one of these two extremes is optimal: $Q^{n\le \bar{n}}_{\scriptsize{\textrm{ODP}}}(\eta,\nth) = Q^{n\le \bar{n}}_{\scriptsize{\textrm{DP}}}(\eta,\nth)$ for $\eta \ge \eta^{\star}(\nth,n)$ and $Q^{n\le \bar{n}}_{\scriptsize{\textrm{ODP}}}(\eta,\nth)=Q^{n\le \bar{n}}_{\scriptsize{\textrm{IDP}}}(\eta,\nth)$ otherwise for some $\eta^{\star}(\nth,n)$. In $\bar{n}\rightarrow \infty$ limit, $\lim_{\bar{n}\rightarrow\infty }\eta^{\star}(\nth,\bar{n}) =1$ and thus $Q^{n\le \bar{n}}_{\scriptsize{\textrm{ODP}}}(\eta,\nth)=Q^{n\le \bar{n}}_{\scriptsize{\textrm{IDP}}}(\eta,\nth)$ for all $\eta\in [0,1]$, consistent with Eq.\eqref{eq:IDP better than DP in energy-unconstrained}. We refer to Ref. \cite{Sharma2017} (e.g., Fig. 6 therein) for a more comprehensive comparison of existing upper bounds, including approximate degradability bounds \cite{Sutter2017}. 


Finally, the following upper bound of energy-unconstrained quantum capacity of the Gaussian random displacement channel can be derived from the data-processing argument combined with Lemma \ref{lemma:loss plus amplification is displacement pre-amplification} (specialized to $\nth =0$): 
\begin{equation}
Q_{\reg}(\mathcal{N}_{B_{2}}[\sigma^{2}]) \le \max \Big{[} \log\Big{(} \frac{1-\sigma^{2}}{\sigma^{2}} \Big{)} ,0 \Big{]}. \label{eq:improved upper bound of Gaussian random displacement channel}
\end{equation}
Note that this bound is strictly tighter than the one in Eq.\eqref{eq:quantum capacity upper bound displacement channel less tight}.  

\section{Gottesman-Kitaev-Preskill codes}
\label{section:GKP codes}

We now aim to find an explicit encoding and decoding strategy which achieves the quantum capacity of Gaussian loss channels. It is very important to realize that the coherent information of bosonic pure-loss channel being maximized by a Gaussian state (i.e., thermal state; Theorem \ref{theorem:thermal optimizer is enough for bosonic pure loss channel capacity}) does not imply that Gaussian encoding and decoding is sufficient to achieve the quantum capacity of bosonic pure-loss channel. In fact, it was proven that entanglement distillation of Gaussian states with Gaussian operation is impossible \cite{Eisert2002}. Due to the close relation between entanglement distillation and quantum state transmission via the quantum teleportation protocol, this implies that Gaussian encoding and decoding can never achieve a non-vanishing quantum communication rate for Gaussian channels. In other words, quantum error correction of a Gaussian error is impossible if the states and operations are restricted to be Gaussian \cite{Niset2009}. Therefore, non-Gaussian resources are necessary in order to achieve reliable quantum information transmission through noisy Gaussian channels. 

In this section, we show that the ability to prepare a Gottesman-Kitaev-Preskill (GKP) state, which is a non-Gaussian state, is sufficient to establish non-vanishing quantum communication rate through Gaussian loss channels. In particular, we prove that the GKP codes achieve the quantum capacity of Gaussian loss channels up to at most a constant gap from the upper bound given in Theorem \ref{theorem:improved upper bound of quantum capacity energy unconstrained}, in the energy-unconstrained scenario.  

\subsection{One-mode square lattice GKP code}
\label{subsection:One-mode square lattice GKP code}

The one-mode square lattice GKP code is the simplest class of the GKP codes \cite{Gottesman2001}. The code space $\mathcal{C}_{\sq}^{[d]} \subset \mathcal{H}$ is defined as $\mathcal{C}_{\sq}^{[d]} = \lbrace |\psi\rangle \,|\, \hat{S}^{[d]}_{\sq,q}|\psi\rangle=|\psi\rangle,\hat{S}^{[d]}_{\sq,p}|\psi\rangle=|\psi\rangle \rbrace$, where 
\begin{align}
\hat{S}^{[d]}_{\sq,q} &\equiv  \exp(i\hat{q}\sqrt{2\pi d} ) = \hat{D}(i\sqrt{\pi d}), 
\nonumber\\
\hat{S}^{[d]}_{\sq,p} &\equiv  \exp(-i\hat{p}\sqrt{2\pi d}) = \hat{D}(\sqrt{\pi d}), \label{eq:one mode square lattice GKP stabilizers}
\end{align}
are the stabilizers and $d$ is the dimension of the code space, i.e., $\textrm{dim}(\mathcal{C}_{\sq}^{[d]}) = d$. Although simultaneous measurement of position and momentum is impossible (i.e., $[\hat{q},\hat{p}]=i\neq 0$), they can nevertheless be measured simultaneously in modulo $\sqrt{2\pi/d}$ through the stabilizer measurement, since the two stabilizers commute with each other $\hat{S}^{[d]}_{\sq,q}\hat{S}^{[d]}_{\sq,p} = \hat{S}^{[d]}_{\sq,p}\hat{S}^{[d]}_{\sq,q}e^{i 2\pi d} = \hat{S}^{[d]}_{\sq,p}\hat{S}^{[d]}_{\sq,q}$. Thus, one-mode square lattice GKP codes can correct any displacement errors in the square unit cell $|\Delta q|,|\Delta p| < \sqrt{\pi/(2d)}$. 

Square lattice GKP code states are explicitly given by
\begin{equation}
|\mu^{\sq}_{L}\rangle \propto \sum_{n=-\infty}^{\infty} |\hat{q} = (dn+\mu)\sqrt{2\pi /d}\rangle 
\end{equation} 
in the computational basis, where $\mu=0,\cdots,d-1$ and $|\hat{q}=q_{0}\rangle$ is an eigenstate of the position operator $\hat{q}$ localized at $q=q_{0}$, which is an infinitely squeezed state. Phase rotation by angle $\theta=-\pi/2$ implements the basis transformation from $|\mu^{\sq}_{L}\rangle$ to
\begin{equation}
|\bar{\mu}^{\sq}_{L}\rangle \propto \sum_{n=-\infty}^{\infty} |\hat{p} = (dn-\mu)\sqrt{2\pi /d}\rangle,  
\end{equation}
where $|\hat{p}=p_{0}\rangle$ is the eigenstate of the momentum $\hat{p}$ at $p=p_{0}$. Note that $|\bar{\mu}^{\sq}_{L}\rangle = \sqrt{1/d}\sum_{\nu=0}^{d-1} \exp(-i2\pi\mu\nu/d) |\nu^{\sq}_{L}\rangle$. 

Since GKP code states are superpositions of infinitely many infinitely squeezed states, their average photon number diverge. It is possible, however, to construct a finite energy GKP states: Replace infinitely squeezed states by finitely squeezed ones, and introduce an overall Gaussian envelop (see Eq.(31) in \cite{Gottesman2001}). Our brief new observation is that this can be concisely realized by $|\mu_{L}^{\Delta}\rangle \propto \exp(-\Delta^{2}\hat{n})|\mu_{L}\rangle$ (which is an equivalent form of Eqs.(40),(41) in \cite{Gottesman2001}). The non-unitary operation $\exp(-\Delta^{2}\hat{n})$ can be physically implemented by passing the input state through a beam splitter with vacuum ancilla, and then by counting photon number and post-selecting the vacuum click at the out-going idler port. As a result, large enough GKP states of any shape can be carved into smaller ones with Gaussian envelop, where the size of the resulting GKP states (i.e., $\Delta$) can be modulated by transmissivity of the beam splitter. 

In Fig. \ref{fig:GKP one mode combined}, we plot the Wigner function of maximally mixed code state $\hat{\rho}_{\mathcal{C}}=\hat{P}_{\mathcal{C}}/d$ of the square lattice and hexagonal lattice (see subsection \ref{subsection:One-mode hexagonal lattice GKP code}) GKP code states for $d=2$ and $\bar{n}=3$, where $\hat{P}_{\mathcal{C}}$ is the projection operator to the code space $\mathcal{C}$. In the remainder of this section, we restrict ourselves to the properties of energy-unconstrained ideal GKP states. 

\begin{figure}[!t]
\centering
\subfloat{\includegraphics[width=1.4in]{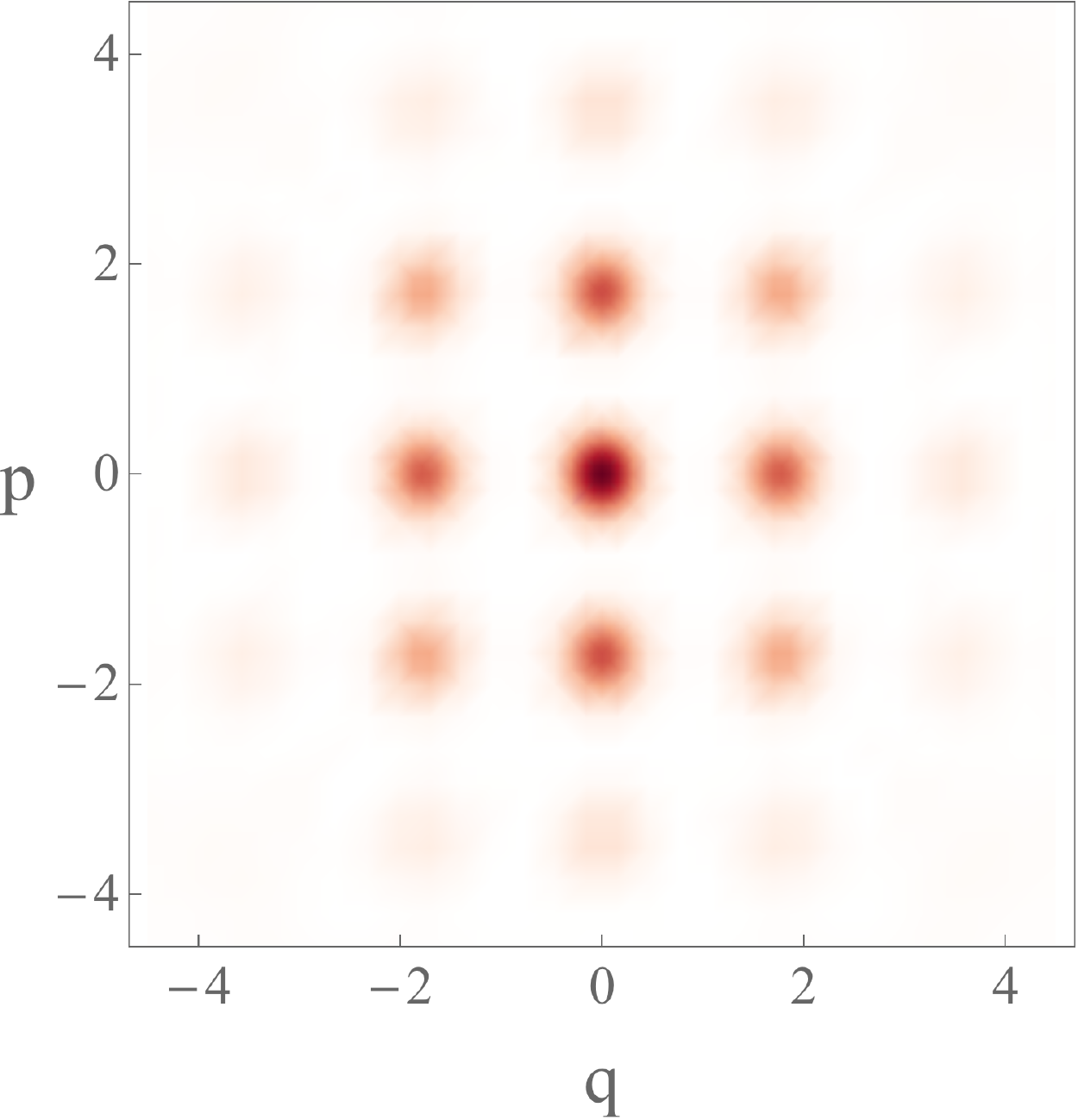}%
\label{fig:GKP square}}
\hfil
\subfloat{\includegraphics[width=1.4in]{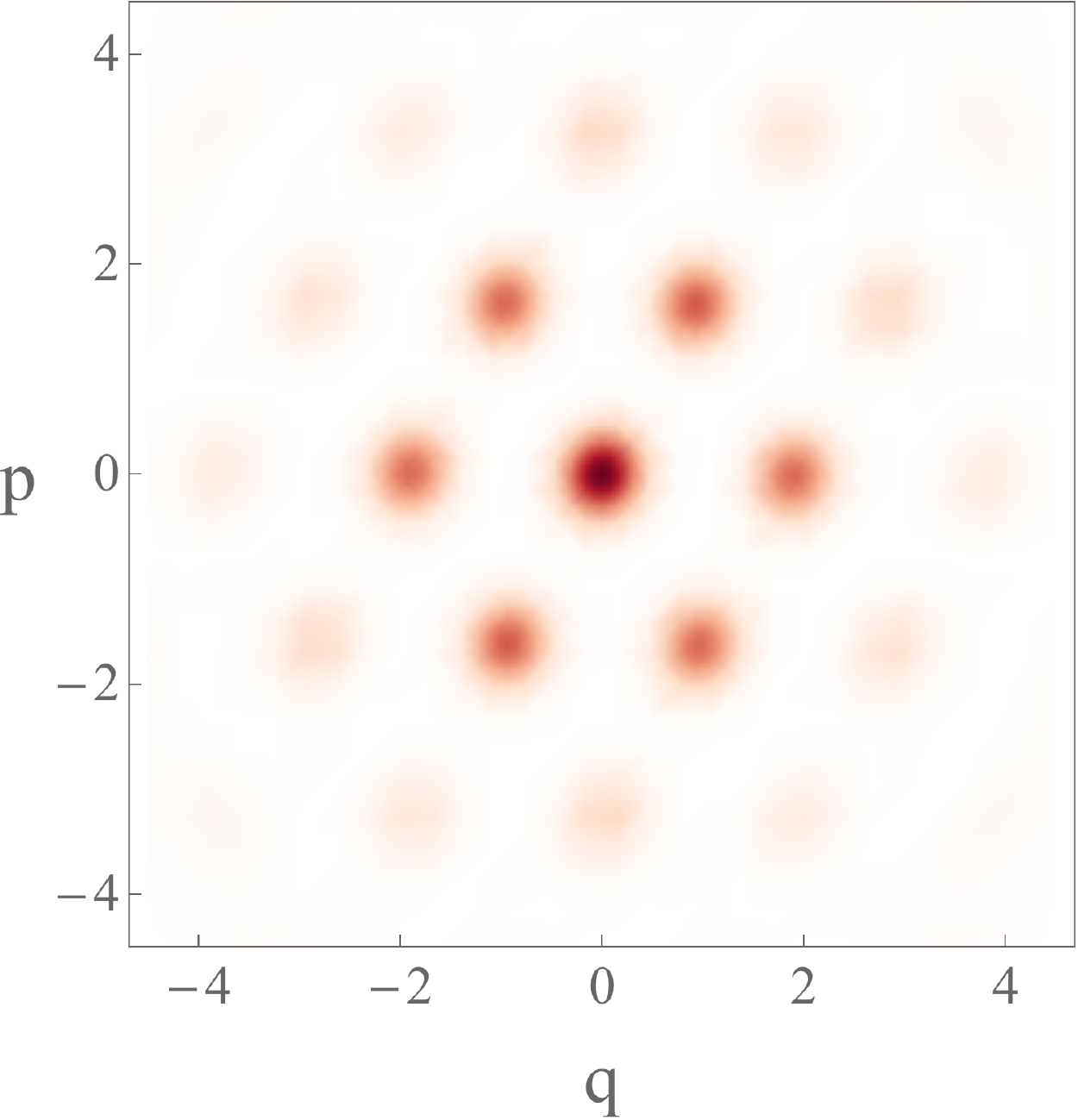}%
\label{fig:GKP hexagonal}}
\hfil
\subfloat{\raisebox{2.8ex}{\includegraphics[width=0.25in]{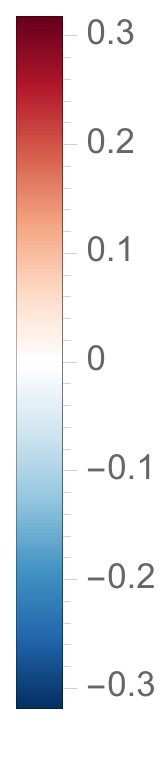}}
\label{fig:GKP PlotLegends}}
\caption{Wigner function of maximally mixed code states of the one-mode square lattice (left; $\mathcal{C}_{\sqtiny}^{[d]}$) and hexagonal lattice (right; $\mathcal{C}_{\hextiny}^{[d]}$) GKP code with $d=2$ and average photon number $\bar{n}=3$.}
\label{fig:GKP one mode combined}
\end{figure}

Universal gate set of the square lattice GKP code is given as follows: Logical Pauli gates (i.e., Clifford gates of hierarchy 1) are given by $\hat{Z}^{\sq}_{L} = (\hat{S}^{[d]}_{\sq,q})^{\frac{1}{d}} = \exp(i \hat{q}\sqrt{ 2\pi /d})$ and $\hat{X}^{\sq}_{L} =  (\hat{S}^{[d]}_{\sq,p})^{\frac{1}{d}}  = \exp(-i\hat{p} \sqrt{ 2\pi /d } )$, i.e., displacement operators, yielding $\hat{Z}^{\sq}_{L}|\mu^{\sq}_{L}\rangle = \exp(i 2\pi\mu /d )|\mu^{\sq}_{L}\rangle$, $\hat{X}^{\sq}_{L}|\mu^{\sq}_{L}\rangle= |(\mu\oplus 1)^{\sq}_{L}\rangle$, where $\oplus$ is the summation in modulo $d$. (The notion of Clifford hierarchy is due to \cite{Gottesman1999}.) Any logical Clifford gates of hierarchy 2 can be implemented by Gaussian unitary operations \cite{Gottesman2001}. For example, logical Hadamard gate can be implemented by phase rotation $\exp(-i \pi \hat{a}^{\dagger}\hat{a} / 2)$, as explained above. Logical CNOT gate can be realized by the SUM gate: $\exp(-i\hat{p}_{1}\hat{q}_{2})|\mu^{\sq}_{L}\rangle|\nu^{\sq}_{L}\rangle = |(\mu\oplus \nu)^{\sq}_{L}\rangle|\nu^{\sq}_{L}\rangle$. In $d=2$ case, we observe that $\exp(i\hat{q}^{4}/(4\pi))$ is the logical T gate (a Clifford gate of hierarchy 3), i.e., $\exp(i\hat{q}^{4}/(4\pi))|\mu^{\sq}_{L}\rangle = \exp(i\pi\mu/4)|\mu^{\sq}_{L}\rangle$, completing the gate set for universal quantum computation.   

Encoding of arbitrary logical states can be achieved as follows: Let $|\psi\rangle = \sum_{\mu=0}^{d-1}c_{\mu}|\mu\rangle$ be an arbitrary state of a physical qudit and assume that one of the GKP logical states (e.g., $|0^{\sq}_{L}\rangle$) is available. Also, assume the gate $\sum_{\nu=0}^{d-1}  (\hat{X}^{\sq}_{L})^{\nu} \otimes |\nu\rangle\langle \nu |$, generated by a Hamiltonian of the form $\hat{H}\propto i(\hat{a}_{1}-\hat{a}_{1}^{\dagger})\otimes \hat{a}_{2}^{\dagger}\hat{a}_{2}$, can be implemented, where $\hat{a}_{1}$ is a bosonic annihilation operator associated with the encoded mode and $\hat{a}_{2} \equiv \sum_{\mu=0}^{d-1}\sqrt{\mu}|\mu-1\rangle\langle \mu|$, associated with the physical qudit. Then, we get 
\begin{equation}
\Big{[}\sum_{\nu=0}^{d-1}  (\hat{X}^{\sq}_{L})^{\nu} \otimes |\nu\rangle\langle \nu |\Big{]} (|0^{\sq}_{L}\rangle\otimes |\psi\rangle) = \sum_{\mu=0}^{d-1}c_{\mu}|\mu^{\sq}_{L}\rangle\otimes |\mu\rangle. 
\end{equation} 
Upon a unitary Fourier gate on the physical qudit, $|\mu\rangle \rightarrow \sqrt{1/d} \sum_{\nu=0}^{d-1} \exp(-i2\pi\mu\nu/d) |\nu\rangle$, this state becomes $\sum_{\mu,\nu=0}^{d-1}c_{\mu}\exp(-i 2\pi \mu\nu/d)|\mu^{\sq}_{L}\rangle\otimes |\nu\rangle$. If one then measures the physical qudit in computational basis $\lbrace|0\rangle, \cdots,|d-1\rangle \rbrace$ and the measurement outcome is $|\nu\rangle$, the oscillator state is collapsed into $|\psi^{\sq}_{L}(\nu)\rangle = \sum_{\mu=0}^{d-1}c_{\mu}\exp( - i 2\pi \mu\nu/d)|\mu^{\sq}_{L}\rangle$. Applying $(\hat{Z}^{\sq}_{L})^{\nu}$, we finally obtain the desired encoded logical state
\begin{equation}
(\hat{Z}^{\sq}_{L})^{\nu} |\psi^{\sq}_{L}(\nu)\rangle = \sum_{\mu=0}^{d-1}c_{\mu}|\mu^{\sq}_{L}\rangle. \label{eq:encoding one mode GKP as desired}
\end{equation}

If the encoded GKP states are sent though the Gaussian random displacement channel, decoding can be achieved by measuring two stabilizers $\hat{S}_{\sq,q}^{[d]},\hat{S}_{\sq,p}^{[d]}$ given in Eq.\eqref{eq:one mode square lattice GKP stabilizers}, providing protection against the channel noise as specified in the text below Eq.\eqref{eq:one mode square lattice GKP stabilizers}. Given that fresh supply of ancillary GKP state is available, such stabilizer measurements can be achieved by Gaussian operations and homodyne detection: Assume we have an ancillary oscillator mode prepared at $|\bar{0}^{\sq}_{L}\rangle$ (i.e., $q_{2}\equiv 0\,\, \textrm{mod}\sqrt{2\pi/d}$). Upon the SUM gate $\exp(-i\hat{q}_{1}\hat{p}_{2})$, the ancillary position operator $\hat{q}_{2}$ is transformed into $\hat{q}_{2}+\hat{q}_{1}$ while the system's position operator remains at $\hat{q}_{1}$. We can then perform homodyne measurement of $\hat{q}_{2}$, and if the measurement outcome is $q_{2}^{M}(=q_{1}+q_{2})$, we can extract the value of $\hat{q}_{1}$ in modulo $\sqrt{2\pi/d}$, i.e., $q_{1} \equiv q_{2}^{M} \,\,\textrm{mod}\sqrt{2\pi/d}$, hence measuring $\hat{S}_{\sq,q}^{[d]}$. The other stabilizer $\hat{S}_{\sq,p}^{[d]}$ can also be measured in a similar way using an initial ancillary GKP state $|0^{\sq}_{L}\rangle$ (i.e., $p_{2}\equiv 0\,\, \textrm{mod}\sqrt{2\pi/d}$), SUM gate with control and target exchanged $\exp(-i\hat{p}_{1}\hat{q}_{2})$ and homodyne measurement of $\hat{p}_{2}$ (see the text below Eq.(105) in \cite{Gottesman2001}). From the extracted values of position and momentum in modulo $\sqrt{2\pi/d}$ (i.e., $q_{2}^{M},p_{2}^{M}$), one can infer $\Delta q = q_{2}^{M\star}$, $\Delta p = p_{2}^{M\star}$, where $q_{2}^{M\star},p_{2}^{M\star}$ are defined such that $q_{2}^{M}\equiv q_{2}^{M\star}\,\,\textrm{mod}\sqrt{2\pi/d}$, $p_{2}^{M}\equiv p_{2}^{M\star}\,\,\textrm{mod}\sqrt{2\pi/d}$ and $|q_{2}^{M\star}|,|p_{2}^{M\star}|\le \sqrt{\pi/(2d)}$. One can then correct such errors by counter displacement $\exp(i\hat{p}_{1}q_{2}^{M\star})$ and $\exp(-i\hat{q}_{1}p_{2}^{M\star})$. 

In both encoding and decoding, the most non-trivial task is to prepare a GKP state (e.g., $|0^{\sq}_{L}\rangle$) from an arbitrary non-GKP state. There have been many proposals to achieve such a goal in various experimental platforms \cite{Travaglione2002,Pirandola2004,Pirandola2006,Vasconcelos2010,
Terhal2016,Motes2017,Weigand2017}, including the one in the original paper \cite{Gottesman2001}. In particular, the proposal in \cite{Terhal2016} is based on the idea of phase estimation of unitary operators \cite{Kitaev1996}, which allows the following representation of the GKP states (equivalent to Eq.(29) in \cite{Gottesman2001}) 
\begin{align}
|\mu^{\sq}_{L}\rangle &\propto  (\hat{X}^{\sq}_{L})^{\mu} \sum_{n_{1},n_{2}=-\infty}^{\infty}  (\hat{S}_{\sq,p}^{[d]})^{n_{1}} (\hat{Z}^{\sq}_{L})^{n_{2}} |\phi_{0}\rangle 
\nonumber\\
&= \sum_{n_{1},n_{2}=-\infty}^{\infty}  e^{ -i\hat{p}\sqrt{\frac{2\pi}{d}} (dn_{1}+\mu) }e^{ i\hat{q}\sqrt{\frac{2\pi}{d}} n_{2} } |\phi_{0}\rangle. \label{eq:representation of GKP state from phase estimation intuition} 
\end{align}
Here, $\sum_{n_{1}=-\infty}^{\infty} (\hat{S}_{p}^{[d]})^{n_{1}}$ and $\sum_{n_{2}=-\infty}^{\infty}(\hat{Z}_{L})^{n_{2}}$ can be understood as the projection operator associated with the phase estimation (and correction) of $\hat{S}_{p}^{[d]}$ and $\hat{Z}_{L}$, which enforce $\hat{S}_{p}^{[d]} = \hat{Z}_{L} = 1$, hence $|0^{\sq}_{L}\rangle$. The last operation $(\hat{X}^{\sq}_{L})^{\mu}$ then transforms $|0^{\sq}_{L}\rangle$ into $|\mu^{\sq}_{L}\rangle$. Note that Eq.\eqref{eq:representation of GKP state from phase estimation intuition} is valid for any input state $|\phi_{0}\rangle$ with non-zero overlap with $|0^{\sq}_{L}\rangle$ which, for example, can be chosen to be the vacuum state. In this case, Eq.\eqref{eq:representation of GKP state from phase estimation intuition} can be understood as a superposition of coherent states $|\alpha\rangle$ in a $2$-dimensional square lattice $\alpha = \sqrt{\frac{\pi}{d}}(dn_{1}+\mu+in_{2})$, where $n_{1},n_{2}$ are integers. 

We remark that the scheme in \cite{Travaglione2002} implements Eq.\eqref{eq:representation of GKP state from phase estimation intuition} by post-selection, whereas \cite{Terhal2016} does so deterministically. The former protocol is within the reach of near-term technology of trapped ion systems, and some preliminary experimental progress has been made \cite{Fluhmann2017}.  In principle, the latter protocol can be realized in circuit quantum electrodynamics systems by, e.g., extending the quantum non-demolition measurement techniques used in \cite{Sun2014}. 

Based on the numerically optimized decoding operation, we earlier showed that the GKP codes offer excellent protection against bosonic pure-loss channel $\mathcal{N}[\eta,0]$ as well, despite not being specifically designed for such a purpose \cite{Albert2017}. In addition, by using sub-optical decoding (i.e., quantum-limited amplification $\mathcal{A}[1/\eta]$ followed by the conventional GKP decoding as described above), we can prove that the logical error probability scales as 
\begin{equation}
\epsilon_{L}^{\sq} \sim \exp\Big{[} - \frac{\pi}{4d}\Big{(}\frac{\eta}{1-\eta}\Big{)} \Big{]}, \label{eq:logical error probability one mode GKP sub-optimal}
\end{equation}
(cf, Eq.(7.24) in \cite{Albert2017}) which vanishes non-analytically in the limit of perfect transmission $\eta \rightarrow 1$. In the following subsection, we will explain that it is possible to improve the constant prefactor from $\frac{\pi}{4}$ to $\frac{\pi}{2\sqrt{3}}$ by using hexagonal lattice GKP code instead of square lattice.  

\subsection{One-mode hexagonal lattice GKP code}
\label{subsection:One-mode hexagonal lattice GKP code}

Code space of one-mode hexagonal lattice GKP code is defined by $\mathcal{C}_{\scriptsize{\textrm{hex}}}^{[d]}\equiv \lbrace |\psi\rangle \,|\, \hat{S}^{[d]}_{\hex,q}|\psi\rangle=|\psi\rangle,\,\,\hat{S}^{[d]}_{\hex,p}|\psi\rangle=|\psi\rangle \rbrace$, where the stabilizers are given by 
\begin{align} 
\hat{S}^{[d]}_{\hex,q} &\equiv  \exp[i\sqrt{2\pi d} (S_{11}\hat{q}+S_{21}\hat{p}) ], 
\nonumber\\
\hat{S}^{[d]}_{\hex,p} &\equiv  \exp[-i\sqrt{2\pi d} (S_{12}\hat{q}+S_{22}\hat{p}) ], \label{eq:one mode hexagonal lattice GKP stabilizers}
\end{align}
and $S_{ij}$ are matrix elements the symplectic matrix 
\begin{equation}
\mathbf{S}_{\hex}= \begin{pmatrix}
S_{11} & S_{12}   \\
S_{21} & S_{22}
\end{pmatrix} = \Big{(}\frac{2}{\sqrt{3}}\Big{)}^{\frac{1}{2}} \begin{pmatrix}
1 & \frac{1}{2}  \\
0 & \frac{\sqrt{3}}{2}  
\end{pmatrix}. \label{eq:one mode hexagonal lattice matrix}
\end{equation}
Note that stabilizers of the square lattice GKP code are Eq.\eqref{eq:one mode hexagonal lattice GKP stabilizers} with $\mathbf{S}_{\sq}=\mathbf{I}_{2}$, where $\mathbf{I}_{2}$ is the $2\times 2$ identity matrix. Thus, hexagonal lattice GKP code states can be generated by applying Gaussian unitary operation, with corresponding symplectic transformation $\mathbf{S}_{\hex}^{-1}$, to the square lattice GKP code states. Similarly as in the case of square lattice GKP code, logical Pauli operators are given by $\hat{Z}_{L}^{\hex} = (\hat{S}^{[d]}_{\hex,q})^{\frac{1}{d}}$ and $\hat{X}_{L}^{\hex} = (\hat{S}^{[d]}_{\hex,p})^{\frac{1}{d}}$. Following the same reasoning to obtain Eq.\eqref{eq:representation of GKP state from phase estimation intuition}, logical states of the hexagonal GKP code are given by 
\begin{align}
&|\mu^{\hex}_{L}\rangle \propto  (\hat{X}^{\hex}_{L})^{\mu} \sum_{n_{1},n_{2}=-\infty}^{\infty}  (\hat{S}_{\hex,p}^{[d]})^{n_{1}} (\hat{Z}^{\hex}_{L})^{n_{2}} |\phi_{0}\rangle 
\nonumber\\
&\,\,= \sum_{n_{1},n_{2}=-\infty}^{\infty}  e^{ -i(\frac{1}{2}\hat{q}+\frac{\sqrt{3}}{2}\hat{p})\sqrt{\frac{4\pi}{\sqrt{3}d}} (dn_{1}+\mu) }   e^{ i\hat{q}\sqrt{\frac{4\pi}{\sqrt{3}d}} n_{2} } |\phi_{0}\rangle, 
\label{eq:representation of GKP state from phase estimation intuition hexagonal} 
\end{align}
in the computational basis, where $|\phi_{0}\rangle$ is an arbitrary state with non-zero overlap with $|0_{L}^{\hex}\rangle$. If $|\phi_{0}\rangle$ is chosen to be the vacuum state, $|\mu_{L}^{\hex}\rangle$ is a superposition of coherent states $|\alpha\rangle$ in a $2$-dimensional hexagonal lattice $\alpha = \sqrt{\frac{2\pi}{\sqrt{3}d}} ((\frac{\sqrt{3}}{2} - \frac{i}{2})(dn_{1}+\mu) + in_{2} )$, where $n_{1},n_{2}$ are integers. For visualization of the one-mode hexagonal lattice GKP code space, see the right panel of Fig. \ref{fig:GKP one mode combined}.  

Note that since $\mathbf{S}_{\sq}$ and $\mathbf{S}_{\hex}$ are symplectic matrices (see Eq.\eqref{eq:symplectic matrix defining property}), the lattice generated by each of them has a unit cell with unit area: $\textrm{det}(\mathbf{S}_{\sq}) = \textrm{det}(\mathbf{S}_{\hex})=1$. For the square lattice, the minimum distance between different lattice points is given by $d_{\min}^{\sq}=1$, and for hexagonal lattice, $d_{\min}^{\hex} = (2/\sqrt{3})^{1/2} > d_{\min}^{\sq}$. Thus, square lattice GKP code can correct displacement errors within the radius $r\equiv \sqrt{|\Delta q|^{2}+|\Delta p|^{2}} \le \sqrt{\frac{\pi}{2d}}$, whereas hexagonal lattice GKP code does so for $r \le \sqrt{\frac{\pi}{\sqrt{3}d}}$. This leads to the following logical error probability
\begin{equation}
\epsilon_{L}^{\hex} \sim \exp\Big{[} - \frac{\pi}{2\sqrt{3}d}\Big{(}\frac{\eta}{1-\eta}\Big{)} \Big{]}, \label{eq:logical error probability one mode GKP sub-optimal hexagonal}
\end{equation}    
for the hexagonal lattice GKP code against bosonic pure-loss channel. Note that the hexagonal lattice allows the densest sphere packing in $2$-dimensional Euclidean space \cite{Fejes1942}. 

\subsection{Multi-mode symplectic dual lattice GKP codes} 

Generalizing Eq.\eqref{eq:one mode hexagonal lattice GKP stabilizers}, one can consider the following stabilizers for $N$-mode GKP code space, encoding $d$ logical states per mode: 
\begin{equation}
\hat{S}^{[d]}_{k} \equiv \exp [i \sqrt{2\pi d} (-1)^{k+1} (\mathbf{\hat{x}}^{T}\mathbf{S})_{k}], \label{eq:stablizers GKP code multi-mode general symplectic matrix}
\end{equation}
for $k\in\lbrace 1,\cdots,2N \rbrace$. Here, $\mathbf{\hat{x}} = (\hat{x}_{1},\cdots,\hat{x}_{2N})^{T} = (\hat{q}_{1},\hat{p}_{1},\cdots,\hat{q}_{N},\hat{p}_{N})^{T}$ are the quadrature operators of $N$ oscillator modes and $\mathbf{S}$ is a $2N\times 2N$ symplectic matrix (see also \cite{Gottesman2001,Harrington2001}). This more general type of GKP code states can be generated by applying Gaussian operation with corresponding symplectic transformation $\mathbf{S}^{-1}$ to the $N$ copies of one-mode square lattice GKP states, i.e., $|\vec{\mu}^{\,\mathbf{S}}_{L}\rangle = \hat{U}_{\mathbf{S}^{-1}}|\vec{\mu}^{\,\sq}_{L}\rangle$, where $\vec{\mu}=(\mu^{1},\cdots,\mu^{N})$, $|\vec{\mu}^{\,\sq}_{L}\rangle = |\mu^{1}_{L}\rangle\otimes \cdots\otimes |\mu^{N}_{L}\rangle$ and $\mu^{k}\in\lbrace0,\cdots,d-1\rbrace$.

\subsection{Achievable quantum communication rate of the GKP codes}
\label{subsection:GKP code rate}

Similar to the one-mode case, it is possible to increase the correctable radius of displacement in $N$-mode case, by choosing a $2N$-dimensional symplectic lattice allowing more efficient sphere packing than the square lattice. It is known that there exists a $2N$-dimensional lattice in Euclidean space allowing $d_{\min}\ge \sqrt{N/(\pi e)}$ \cite{Hlawka1943} and a stronger statement was proven in \cite{Buser1994} that the same holds also for symplectic lattices. Choosing such a lattice to define the GKP code, one can correct all displacement error within radius $r\le   \sqrt{N/( 2 e d)}$. In the Gaussian random displacement channel $\mathcal{N}_{B_{2}}[\sigma^{2}]$, the probability of displacement with radius larger than $ \sqrt{2N}\sigma$ occurring vanishes in the limit of infinitely many modes $N\rightarrow\infty$. Thus, if $\sqrt{N/( 2 e d)}\ge  \sqrt{2N}\sigma$ is satisfied, i.e., 
\begin{equation}
d\le d_{\sigma} \equiv  \frac{1}{4e\sigma^{2}}, 
\end{equation}
encoded information can be transmitted faithfully with vanishing decoding error probability. Then, it follows that the communication rate $R=\log \lfloor d_{\sigma} \rfloor = \log \lfloor \frac{1}{4e\sigma^{2}}\rfloor$ can be achieved for the Gaussian random displacement channel (see Eq.(55) in \cite{Harrington2001}; floor function is due to the fact that $d$ can only be an integer).  

\begin{figure*}[!t]
\centering
\subfloat{\includegraphics[width=3.2in]{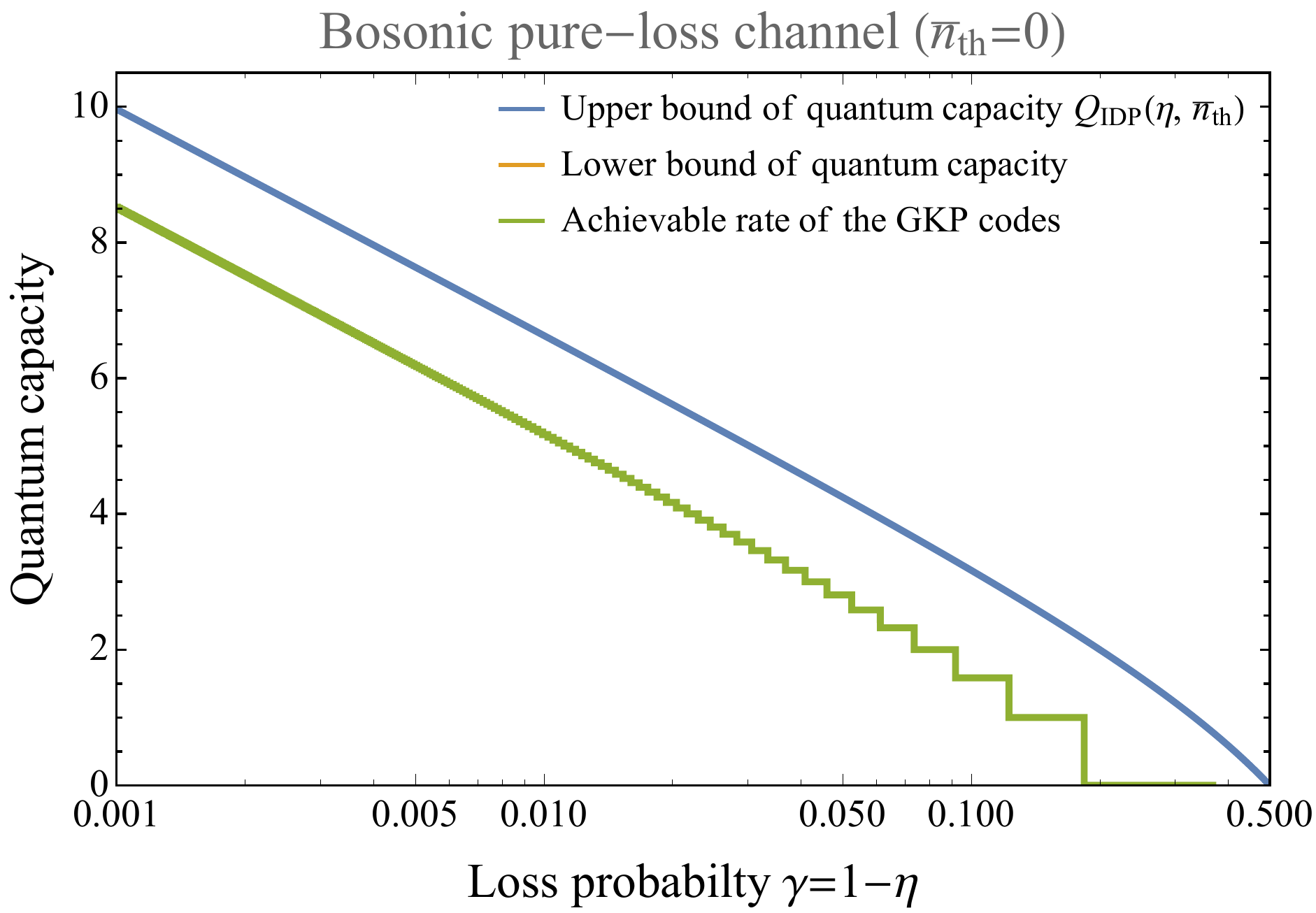}%
\label{fig:GKP rate compared with pure loss channel capacity}}
\hfil
\subfloat{\includegraphics[width=3.2in]{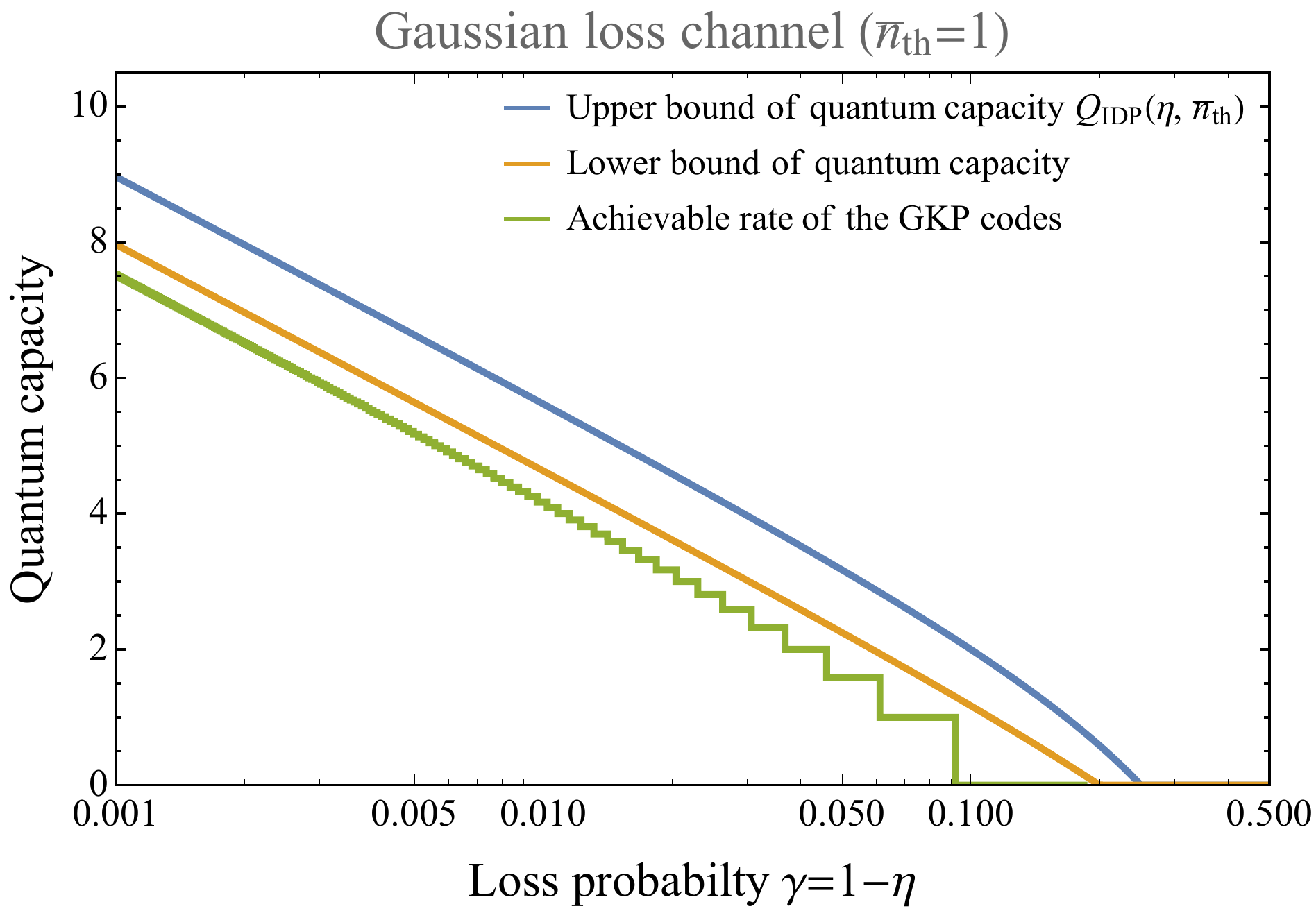}%
\label{fig:GKP rate compared with thermal loss channel capacity}}
\caption{Achievable quantum communication rate of the GKP codes (green) compared with lower (orange) and upper (blue) bounds of quantum capacity of bosonic pure-loss channel $\mathcal{N}[\eta,\nthtiny = 0]$ (left) and Gaussian loss channel $\mathcal{N}[\eta,\nthtiny = 1]$ (right). Orange line in the left panel is hidden behind the blue line, since for the bosonic pure-loss channel, lower and upper bounds coincide with each other and the quantum capacity is completely known.}
\label{fig:Achievable rate of GKP codes compared with capacity}
\end{figure*}

Note that the above estimation is overly conservative since we did not take into account the correctable displacement outside the correctable sphere. With an improved estimation of the decoding error probability, the following statement was ultimately established in \cite{Harrington2001}:  

\begin{lemma}[Eq.(66) in \cite{Harrington2001}]  
Let $\mathcal{N}_{B_{2}}[\sigma^{2}]$ be a Gaussian random displacement channel as defined in Definition \ref{definition:Gaussian displacement channel}. Also, let $\mathcal{C}_{\mathbf{S}}^{[d]}$ be the $N$-mode GKP code space defined by the stabilizers in Eq.\eqref{eq:stablizers GKP code multi-mode general symplectic matrix}. Then, there exists a symplectic lattice generated by $\mathbf{S}$ such that the GKP code achieves the following rate for the Gaussian random displacement channel $\mathcal{N}_{B_{2}}[\sigma^{2}]$ in $N\rightarrow\infty$ limit.  
\begin{equation}
R = \max \Big{(} \log \Big{\lfloor} \frac{1}{e\sigma^{2}} \Big{\rfloor} , 0  \Big{)}.  
\end{equation}
Here, $\lfloor x \rfloor$ is the floor function, due to the fact that $d$ can only be an integer.  \label{lemma:achievable GKP communication rate for Gaussian random displacement channel}
\end{lemma}

Thus, the GKP codes achieve one-shot coherent information of the Gaussian random displacement channel (cf, Eq.\eqref{eq:quantum capacity Gaussian displacement channel lower bound}) in modulo floor function. Based on this, we now establish the following main result: 
\begin{theorem}[Achievable quantum communication rate of the GKP codes for Gaussian loss channel]
Let $\mathcal{N}[\eta,\nth]$ be a Gaussian loss channel as defined in Definition \ref{definition:Gaussian loss channels}. Upon optimal choice of $2N\times 2N$ symplectic matrix $\mathbf{S}=\mathbf{S}^{\star}$, the $N$-mode GKP code as defined in Eq.\eqref{eq:stablizers GKP code multi-mode general symplectic matrix} achieves the following rate for the Gaussian loss channel $\mathcal{N}[\eta,\nth]$ in $N\rightarrow \infty$ limit: 
\begin{equation}
R = \max \Big{(} \log \Big{\lfloor} \frac{1}{e(1-\eta)(\nth+1)} \Big{\rfloor} , 0  \Big{)}, \label{eq:achievable rate of the GKP codes for Gaussian loss channel}
\end{equation}
where $\lfloor x \rfloor$ is the floor function. \label{theorem:achievable rate of the GKP codes for Gaussian loss channels}
\end{theorem}
\begin{proof}
Lemma \ref{lemma:loss plus amplification is displacement pre-amplification} states that the Gaussian loss channel $\mathcal{N}[\eta,\nth]$ can be converted into the Gaussian random displacement channel $\mathcal{N}_{B_{2}}[\tilde{\sigma}^{2}_{\eta,\nthtiny}]$ by quantum-limited amplification $\mathcal{A}[1/\eta]$, where $\tilde{\sigma}^{2}_{\eta,\nthtiny} = (1-\eta)(\nth+1)$. Combining this with Lemma \ref{lemma:achievable GKP communication rate for Gaussian random displacement channel}, Eq.\eqref{eq:achievable rate of the GKP codes for Gaussian loss channel} follows. 
\end{proof}

Recall Theorem \ref{theorem:improved upper bound of quantum capacity energy unconstrained} and note that $Q_{\reg}(\mathcal{N}[\eta,\nth]) \le Q_{\scriptsize{\textrm{IDP}}}(\eta,\nth)$ where
\begin{equation}
Q_{\scriptsize{\textrm{IDP}}}(\eta,\nth) < \max \Big{[}\log\Big{(} \frac{1}{(1-\eta)(\nth+1)} \Big{)}, 0\Big{]}, 
\end{equation}
as derived from Eq.\eqref{eq:improved upper bound of quantum capacity energy unconstrained simplified}. Comparing this with Eq.\eqref{eq:achievable rate of the GKP codes for Gaussian loss channel}, we get $Q_{\reg}(\mathcal{N}[\eta,\nth]) - R \lesssim \log e = 1.44269\cdots $, where $\sim$ is due to the floor function. Thus, GKP codes defined over optimal symplectic lattice achieve quantum capacity of Gaussian loss channels up to at most a constant gap from the upper bound of quantum capacity (see Fig. \ref{fig:Achievable rate of GKP codes compared with capacity} for illustration).

We note that the established rate in Theorem \ref{theorem:achievable rate of the GKP codes for Gaussian loss channels} relies on the existence of a symplectic lattice in higher dimensions satisfying a certain desired condition (see Eqs.(56),(57) in \cite{Harrington2001}). In this regard, we remark that $E_{8}$ lattice and Leech lattice $\Lambda_{24}$ (both symplectic; see appendix of \cite{Buser1994}) were recently shown to support the densest sphere packing in 8 and 24 dimensional Euclidean space, respectively \cite{Viazovska2017,Cohn2017}. 

\section{Biconvex encoding and decoding optimization}
\label{section:Biconvex encoding and decoding optimization}

In this subsection, we consider the energy-constrained scenario by imposing maximum allowed photon number in each mode. Instead of attempting to establish a theoretically rigorous statement as we did in the energy-unconstrained case, we tackle the problem by numerical optimization. Since decoding error probability cannot be suppressed arbitrarily close to zero with only finite number of modes, we aim to find a set of optimal encoding and decoding maps which maximize the \textit{entanglement fidelity}. The main reason we chose entanglement fidelity (instead of other measures, e.g., diamond norm) as the fidelity measure is to make the optimization problem more tractable, as maximization of entanglement fidelity can be formulated as a biconvex optimization, which can be tackled by alternating semidefinite programming method.  

\subsection{Choi matrix and superoperator of a quantum channel} 

A quantum channel $\mathcal{A}:\mathcal{L}(\mathcal{H}_{1})\rightarrow\mathcal{L}(\mathcal{H}_{2})$ is a completely positive trace preserving (CPTP) map \cite{Choi1975}, which has one-to-one correspondence with the Choi matrix $\hat{X}_{\mathcal{A}}\in\mathcal{L}(\mathcal{H}_{1}\otimes \mathcal{H}_{2})$. Matrix elements of the Choi matrix are defined as $(\hat{X}_{\mathcal{A}})_{[ij],[i'j']} = \langle j_{\mathcal{H}_{2}}| \mathcal{A}(|i_{\mathcal{H}_{1}}\rangle\langle i'_{\mathcal{H}_{1}}|)|j'_{\mathcal{H}_{2}}\rangle$,  
where $|i_{\mathcal{H}_{1}}\rangle,|i'_{\mathcal{H}_{1}}\rangle$ and $|j_{\mathcal{H}_{2}}\rangle$, $|j'_{\mathcal{H}_{2}}\rangle$ are orthonormal basis of $\mathcal{H}_{1}$ and $\mathcal{H}_{2}$, respectively. Choi matrix $\hat{X}_{\mathcal{A}}$ is hermitian positive semidefinite by definition of complete positivity (CP) of $\mathcal{A}$ \cite{Choi1975}. Also, trace preserving (TP) condition imposes the affine constraint $\textrm{Tr}_{\mathcal{H}_{2}}\hat{X}_{\mathcal{A}} \equiv \sum_{j=0}^{\scriptsize{\textrm{dim}}(\mathcal{H}_{2})-1} (X_{\mathcal{A}})_{[ij],[i'j]} |i_{\mathcal{H}_{1}}\rangle\langle i'_{\mathcal{H}_{1}}|  = \hat{I}_{\mathcal{H}_{1}}$. 

Let $\mathcal{B}:\mathcal{L}(\mathcal{H}_{2})\rightarrow \mathcal{L}(\mathcal{H}_{3})$ be another quantum channel (i.e., a CPTP map) and consider the composite channel $\mathcal{B}\cdot\mathcal{A}:\mathcal{L}(\mathcal{H}_{1})\rightarrow \mathcal{L}(\mathcal{H}_{3})$. To analyze the composite channel, it is convenient to define superoperator $\hat{T}_{\mathcal{A}}$ of a channel $\mathcal{A}$ whose matrix elements are given by $(\hat{T}_{\mathcal{A}})_{jj',ii'}\equiv \langle j_{\mathcal{H}_{2}}| \mathcal{A}(|i_{\mathcal{H}_{1}}\rangle\langle i'_{\mathcal{H}_{1}}|)|j'_{\mathcal{H}_{2}}\rangle = (\hat{X}_{\mathcal{A}})_{[ij],[i'j']}$ \cite{Preskill}: Superoperator of a composite channel is then given by the matrix multiplication of the superoperators of its constituting channels, i.e., $\hat{T}_{\mathcal{B}\cdot\mathcal{A}} = \hat{T}_{\mathcal{B}}\hat{T}_{\mathcal{A}}$.   

\subsection{Entanglement fidelity and Choi matrix}

Let $\mathcal{H}_{n}$ be a Hilbert space of dimension $n$ and $\mathcal{N}:\mathcal{L}(\mathcal{H}_{n})\rightarrow\mathcal{L}(\mathcal{H}_{n})$ be a CPTP map describing a noisy quantum channel. Consider $d$-dimensional ($d\le n$) Hilbert spaces $\mathcal{H}_{0}$ and $\mathcal{H}_{0}'$ and assume the information sender prepared a maximally entangled state in a local noiseless memory: 
\begin{equation}
|\Phi^{+} \rangle \equiv \frac{1}{\sqrt{d}}\sum_{i =0 }^{d-1} |i_{\mathcal{H}_{0}}\rangle|i_{\mathcal{H}_{0}'}\rangle. 
\end{equation}
The sender then encodes half of the entangled state in $\mathcal{H}_{0}$ to the channel input by a CPTP encoding map $\mathcal{E}:\mathcal{L}(\mathcal{H}_{0})\rightarrow \mathcal{L}(\mathcal{H}_{n})$, and then send it through the noisy channel $\mathcal{N}:\mathcal{L}(\mathcal{H}_{n})\rightarrow\mathcal{L}(\mathcal{H}_{n})$. The receiver then maps the received state at the channel output to the local memory by a CPTP decoding map $\mathcal{D}:\mathcal{L}(\mathcal{H}_{n})\rightarrow \mathcal{L}(\mathcal{H}_{0})$. As a result, both parties obtain an approximate entangled state 
\begin{equation}
\hat{\rho} \equiv (\mathcal{D}\cdot\mathcal{N}\cdot\mathcal{E}\otimes \textrm{id}_{\mathcal{H}_{0}'})(|\Phi^{+}\rangle\langle\Phi^{+}|),  \label{eq:approximate entangled state}
\end{equation}     
where $\textrm{id}_{\mathcal{H}_{0}'}:\mathcal{L}(\mathcal{H}_{0}')\rightarrow \mathcal{L}(\mathcal{H}_{0}')$ is the identity map associated with the noiseless local memory at the sender's side. 
Note that $\hat{\rho}$ in Eq.\eqref{eq:approximate entangled state} is explicitly given by 
\begin{equation}
\hat{\rho} = \frac{1}{d}\sum_{i,i'=0}^{d-1} (\hat{X}_{\mathcal{D}\cdot\mathcal{N}\cdot\mathcal{E}})_{[ij],[i'j']} |j_{\mathcal{H}_{0}}\rangle\langle j'_{\mathcal{H}_{0}} | \otimes |i_{\mathcal{H}_{0}'}\rangle\langle i'_{\mathcal{H}_{0}'} | , 
\end{equation}
where $\hat{X}_{\mathcal{D}\cdot\mathcal{N}\cdot\mathcal{E}}\in\mathcal{L}(\mathcal{H}_{0}\otimes \mathcal{H}_{0})$ is the Choi matrix of the composite channel $\mathcal{D}\cdot\mathcal{N}\cdot\mathcal{E}$. Thus, $\hat{\rho}$ is isomorphic to the Choi matrix $\hat{X}_{\mathcal{D}\cdot\mathcal{N}\cdot\mathcal{E}}$.  

Quality of the approximate non-local entangled state $\hat{\rho}$ can be measured by the entanglement fidelity $F_{e}(\hat{\rho})\equiv \langle \Phi^{+}|\hat{\rho}|\Phi^{+}\rangle$ \cite{Schumacher1996B}, where $F_{e}(\hat{\rho}) = 1$ iff $\hat{\rho}$ is a noiseless maximally entangled state $|\Phi^{+}\rangle\langle \Phi^{+}|$. Isomorphism between $\hat{\rho}$ and $\hat{X}_{\mathcal{D}\cdot\mathcal{N}\cdot\mathcal{E}}$ then yields 
\begin{equation}
F_{e}(\hat{\rho})=(1/d^{2})\sum_{i,i'=0}^{d-1}(\hat{X}_{\mathcal{D}\cdot\mathcal{N}\cdot\mathcal{E}})_{[ii],[i'i']} = \frac{1}{d^{2}}\textrm{Tr}[\hat{T}_{\mathcal{D}\cdot\mathcal{N}\cdot\mathcal{E}}], 
\end{equation}
where $\hat{T}_{\mathcal{D}\cdot\mathcal{N}\cdot\mathcal{E}}$ is the superoperator of $\mathcal{D}\cdot\mathcal{N}\cdot\mathcal{E}$. Note that $\hat{T}_{\mathcal{D}\cdot\mathcal{N}\cdot\mathcal{E}} = \hat{T}_{\mathcal{D}} \hat{T}_{\mathcal{N}} \hat{T}_{\mathcal{E}} $ can be decomposed as 
\begin{align} 
&(\hat{T}_{\mathcal{D}\cdot\mathcal{N}\cdot\mathcal{E}})_{jj',ii'}  
\nonumber\\
&\,\,\,=\sum_{k,k',l,l'=0}^{n-1} (\hat{X}_{\mathcal{D}})_{[lj],[l'j']} (\hat{X}_{\mathcal{N}})_{[kl],[k'l']}  (\hat{X}_{\mathcal{E}})_{[ik],[i'k']},  \label{eq:superoperator of composite channel in terms of constituting Choi matrices}
\end{align}
where $\hat{X}_{\mathcal{D}}\in\mathcal{L}(\mathcal{H}_{0}\otimes \mathcal{H}_{n})$, $\hat{X}_{\mathcal{N}}\in\mathcal{L}(\mathcal{H}_{n}\otimes \mathcal{H}_{n})$ and $\hat{X}_{\mathcal{E}}\in\mathcal{L}(\mathcal{H}_{n}\otimes \mathcal{H}_{0})$ are the Choi matrices of decoding, noisy channel, and encoding, respectively, and $i,i',j,j'\in \lbrace 0,\cdots,d-1 \rbrace$. We note that the entanglement fidelity $F_{e}(\hat{\rho})$ is a bi-linear function of $\hat{X}_{\mathcal{E}}$ and $\hat{X}_{\mathcal{E}}$, since $\hat{T}_{\mathcal{D}\cdot\mathcal{N}\cdot\mathcal{E}}$ is bi-linear in $\hat{X}_{\mathcal{E}}$ and $\hat{X}_{\mathcal{E}}$, as can be seen from Eq.\eqref{eq:superoperator of composite channel in terms of constituting Choi matrices}.

To make the bi-linearity more evident, we define a linear map $f_{\mathcal{N}}:\mathcal{L}(\mathcal{H}_{n}\otimes \mathcal{H}_{0})\rightarrow \mathcal{L}(\mathcal{H}_{0}\otimes \mathcal{H}_{n})$ such that 
\begin{equation}
\big{(}f_{\mathcal{N}}(\hat{X})\big{)}_{[l'i'],[li]} \equiv \sum_{k,k'=0}^{n-1} (\hat{X}_{\mathcal{N}})_{[kl],[k'l']}  (\hat{X})_{[ik],[i'k']}, 
\end{equation}
where $l,l'\in\lbrace 0, \cdots,n-1\rbrace$. Entanglement fidelity $F_{e}(\hat{\rho})$ is then given by 
\begin{equation}
F_{e}(\hat{\rho}) = \frac{1}{d^{2}}\textrm{Tr}\Big{[}\hat{X}_{\mathcal{D}} f_{\mathcal{N}}(\hat{X}_{\mathcal{E}}) \Big{]},  
\label{eq:entanglement fidelity in terms of Choi matrices} 
\end{equation}
which is apparently bi-linear in $\hat{X}_{\mathcal{E}}$ and $\hat{X}_{\mathcal{D}}$.

\subsection{Maximization of entanglement fidelity}

\begin{figure*}[!t]
\centering
\includegraphics[width=6.5in]{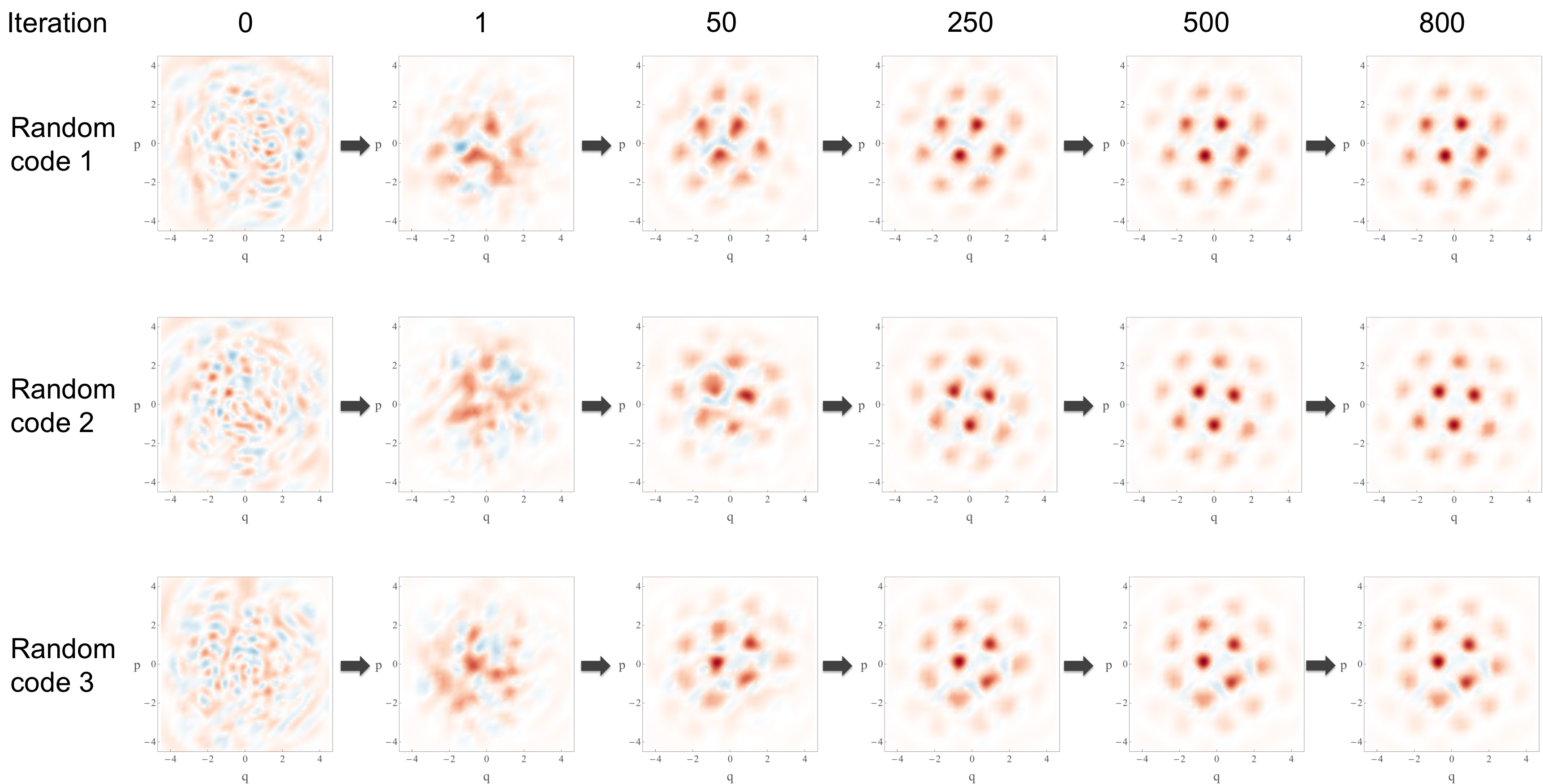}
\caption{Biconvex optimization of the encoding and decoding maps for bosonic pure-loss channel $\mathcal{N}[\eta,0]$ with $\eta=0.9$. We chose $n=20$ and $d=2$ and imposed photon number constraint $\textrm{Tr}[\hat{n}\hat{\rho}_{\mathcal{E}}] \le  \ 3$, where $\hat{\rho}_{\mathcal{E}} = (1/d) \textrm{Tr}_{\mathcal{H}_{0}}\hat{X}_{\mathcal{E}}$ is the maximally mixed code state. First column of each row is the Wigner function of $\hat{\rho}_{\mathcal{E}}$ for a randomly generated encoding map $\mathcal{E}$. From the second to the sixth columns represent the updated code spaces after $1,50,250,500$ and $800$ iterations of alternating semidefinite programming. Color scale is the same as in Fig. \ref{fig:GKP one mode combined}.}
\label{fig:biconvex optimization combined}
\end{figure*}

It is known that finding optimal decoding operation which maximizes the entanglement fidelity is a semidefinite programming, if the encoding map is fixed to $\mathcal{E}=\bar{\mathcal{E}}$ \cite{Reimpell2005,Fletcher2007}:
\begin{alignat}{2}
&\max_{\hat{X}_{\mathcal{D}}}&& \textrm{Tr}[ \hat{X}_{\mathcal{D}}(f_{\mathcal{N}}(\hat{X}_{\bar{\mathcal{E}}}))]
\nonumber\\
&\,\,\,\textrm{s.t.}&& \hat{X}_{\mathcal{D}}=\hat{X}_{\mathcal{D}}^{\dagger}\succeq 0,\,\, \textrm{Tr}_{\mathcal{H}_{0}}\hat{X}_{\mathcal{D}} = \hat{I}_{\mathcal{H}_{n}}, 
\end{alignat} 
where the constrains are due to the CPTP nature of $\mathcal{D}$. Similarly, optimizing encoding while fixing decoding is also a semidefinite programming, and thus the entire problem is a biconvex optimization \cite{Kosut2009}. Here, we further impose an energy constraint to the encoding map while still preserving bi-convexity of the problem. Let $\hat{E}\in \mathcal{L}(\mathcal{H}_{n})$ be an energy observable. We define $\textrm{Tr}_{\mathcal{H}_{n}}[\hat{E}\hat{\rho}_{\mathcal{E}}]$ to be the average energy in the encoded state, where $\hat{\rho}_{\mathcal{E}}\equiv \mathcal{E}(\frac{1}{d}\sum_{i=0}^{d-1}|i_{\mathcal{H}_{0}}\rangle\langle i_{\mathcal{H}_{0}}|)=\frac{1}{d} \textrm{Tr}_{\mathcal{H}_{0}} \hat{X}_{\mathcal{E}}$ is the maximally mixed encoded state. Then, the energy constraint reads $\textrm{Tr}[(\hat{E}\otimes \hat{I}_{\mathcal{H}_{0}}) \hat{X}_{\mathcal{E}}] \le  \bar{E}d$ and we obtain the following biconvex encoding and decoding optimization with energy constraint: 
\begin{alignat}{2}
&\max_{\hat{X}_{\mathcal{E}},\hat{X}_{\mathcal{D}}}&& \textrm{Tr}[\hat{X}_{\mathcal{D}}^{\dagger}f_{\mathcal{N}}(\hat{X}_{\mathcal{E}})], \,\,
\nonumber\\
&\quad \textrm{s.t.}&& \hat{X}_{\mathcal{D}}=\hat{X}_{\mathcal{D}}^{\dagger}\succeq 0,\,\, \textrm{Tr}_{\mathcal{H}_{0}}\hat{X}_{\mathcal{D}} = \hat{I}_{\mathcal{H}_{n}}, 
\nonumber\\
&&& \hat{X}_{\mathcal{E}}=\hat{X}_{\mathcal{E}}^{\dagger}\succeq 0,\,\, \textrm{Tr}_{\mathcal{H}_{n}}\hat{X}_{\mathcal{E}} = \hat{I}_{\mathcal{H}_{0}}, 
\nonumber\\
&&&\quad\textrm{and}\,\,\textrm{Tr}[(\hat{E}\otimes \hat{I}_{\mathcal{H}_{0}}) \hat{X}_{\mathcal{E}}] \le  \bar{E}d , 
\label{eq:biconvex optimization formulation} 
\end{alignat}
where the second line in the constraint is due to the CPTP nature of the encoding map $\mathcal{E}$, and the third line is due to the encoding energy constraint. 

We apply Eq.\eqref{eq:biconvex optimization formulation} to find an optimal qubit-into-oscillator encoding (i.e., $\textrm{dim}(\mathcal{H}_{0})=d=2$) against the bosonic pure-loss channel $\mathcal{N}[\eta,0]$ subject to the photon number constraint $\textrm{Tr}[\hat{n}\hat{\rho}_{\mathcal{E}}] \le \bar{n}$. For practical implementation, since the photon number cannot increase under the bosonic pure-loss channel, we confine the bosonic Hilbert space to a truncated subspace $\mathcal{H}_{n}\equiv\textrm{span}\lbrace |0\rangle,\cdots,|n-1\rangle \rbrace$, 
yielding $\mathcal{N}(\hat{\rho})=\mathcal{N}[\eta,0](\hat{\rho}) = \sum_{\ell =0}^{n-1}\hat{E}_{\ell}\hat{\rho}\hat{E}_{\ell}^{\dagger}$ where $\hat{E}_{\ell}  = \sqrt{(1-\eta)^{\ell}/\ell!} \eta^{\frac{\hat{n}}{2}}\hat{a}^{\ell}$. We choose $n\gg \bar{n}$ to avoid artifacts caused by truncation and solve Eq.\eqref{eq:biconvex optimization formulation} heuristically by alternating encoding and decoding optimization, each solved by semidefinite programming, starting from a random initial encoding map. We generate random initial codes by taking the first $d$ columns of a Haar random $n\times n$ unitary matrix. To solve each SDP sub-problem we used CVX, a package for specifying and solving convex programs \cite{CVX,Grant2008}.

In Fig. \ref{fig:biconvex optimization combined}, we took $\eta=0.9$, $n=20$, $d=2$ and $\bar{n}=3$ and plot the Wigner function of maximally mixed code states of the numerically optimized codes (last column), starting from three different random initial encoding maps (first column). In all instances, the obtained codes are similar to the hexagonal GKP code (in Fig. \ref{fig:GKP one mode combined}), except an overall displacement. We note that the optimized code in second row exhibits the best performance (i.e., $1-F_{\mathcal{N}}^{\star}=0.002092$) among all. Although we do not attempt to make a rigorous claim, this might indicate that the GKP codes defined over the symplectic lattice allowing the most efficient sphere packing may be the optimal encoding for the Gaussian loss channels. In particular, the constant gap in Theorem \ref{theorem:achievable rate of the GKP codes for Gaussian loss channels} may be closed if we assume the optimal decoding, instead of the sub-optimal one involving (noisy) quantum-limited amplification, which we assumed to prove Eq.\eqref{eq:achievable rate of the GKP codes for Gaussian loss channel}. 

We emphasize that the biconvex optimization in Eq.\eqref{eq:biconvex optimization formulation} explores the most general form of encoding maps, including the mixed state encoding. From the numerical optimization, however, we only obtained pure-state encoding (i.e., $\mathcal{E}(\hat{\rho}_{0})=\hat{V}\hat{\rho}_{0}\hat{V}^{\dagger}$, where $\hat{V}:\mathcal{H}_{0}\rightarrow\mathcal{H}_{n}$ is an isometry $\hat{V}^{\dagger}\hat{V}=\hat{I}_{\mathcal{H}_{0}}$) as an optimal solution at all iterations of SDP sub-problem. We also stress that the alternating semidefinite programming method is not guaranteed to yield a global optimal solution, although the obtained hexagonal GKP code in Fig. \ref{fig:biconvex optimization combined} may be a global optimum. In principle, global optimal solution of Eq.\eqref{eq:biconvex optimization formulation} can be deterministically found by global optimization algorithm outlined in \cite{Floudas1990}. To implement the algorithm, however, one should in general solve exponentially many convex sub-problems in the number of complicating variables (responsible for non-convexity of the problem; see \cite{Floudas1990} for details), which is intractable in our application. 

\section{Conclusion}

Exploiting various synthesis and decomposition of Gaussian channels, we provided improved upper bounds on the Gaussian thermal loss channel, both in the energy constrained and unconstrained cases, and the Gaussian displacement channel (Theorems \ref{theorem:improved upper bound of quantum capacity energy unconstrained}, \ref{theorem:improved upper bound of quantum capacity energy constrained}, \ref{theorem:optimized data-processing upper bound} and Eq.\eqref{eq:improved upper bound of Gaussian random displacement channel}). We also established an achievable rate of the GKP codes for the Gaussian loss channels (Theorem \ref{theorem:achievable rate of the GKP codes for Gaussian loss channels}). In the energy-unconstrained case, in particular, we showed that the GKP codes achieve the upper bound of Gaussian loss channel capacity upto at most a constant gap $\simeq \log e = 1.44269\cdots$ in the unit of qubits per channel use. In the energy-constrained case, we formulated a biconvex encoding and decoding optimization problem and solved it via alternating semidefinite programming method. We demonstrated in the one-mode case that hexagonal GKP code emerges as an optimal encoding from Haar random initial codes.     

As was proven in \cite{Barnum2002}, entanglement infidelity of the Petz decoding \cite{Petz1986,Petz1988} (or, transpose decoding) is at most twice as large as the optimal entanglement infidelity. Thus, communication rate achievable by Petz decoding equals to the optimal achievable rate. We leave finding the optimal communication rate of the GKP codes for Gaussian loss channels as an open problem.   

\textit{Note added}: While preparing the manuscript, we became aware of a related work \cite{Rosati2018} in which Lemma \ref{lemma:Gaussian loss channel decomposed into pure loss and amplficiation pre-amplification} and Theorem \ref{theorem:improved upper bound of quantum capacity energy unconstrained} were independently discovered. We also realized our Theorem \ref{theorem:improved upper bound of quantum capacity energy constrained} was independently discovered in \cite{Sharma2017}. We thank the authors of \cite{Sharma2017} for pointing out a mistake we made in Eq.\eqref{eq:new photon number after amplification} in an earlier version of our manuscript. 

\section*{Acknowledgment}

We would like to thank Steven M. Girvin, Barbara Terhal, John Preskill, Steven T. Flammia, Sekhar Tatikonda, Richard Kueng, Linshu Li and Mengzhen Zhang for fruitful discussions. We also thank Mark M. Wilde and Matteo Rosati for useful comments on our manuscript. We acknowledge support from the ARL-CDQI (W911NF-15-2-0067), ARO (W911NF-14-1-0011, W911NF-14-1-0563), ARO MURI (W911NF-16-1-0349 ), AFOSR MURI (FA9550-14-1-0052, FA9550-15-1-0015), NSF (EFMA-1640959), the Alfred P. Sloan Foundation (BR2013-049), and the Packard Foundation (2013-39273). KN acknowledges support through the Korea Foundation for Advanced Studies. VAA acknowledges support from the Walter Burke Institute for Theoretical Physics at Caltech.  

\appendix

\section{Bosonic mode, Gaussian states and Gaussian unitary operations}
\label{section:Bosonic mode, Gaussian states and Gaussian unitary operations}

Let $\mathcal{H}$ denote an infinite-dimensional Hilbert space. Quantum states of $N$ bosonic modes are in the tensor product of $N$ such Hilbert spaces $\mathcal{H}^{\otimes N}$. Each bosonic mode is associated with the annihilation and creation operators $\hat{a}_{k}$ and $\hat{a}^{\dagger}_{k}$, satisfying the bosonic communication relation $[\hat{a}_{i},\hat{a}_{j}] = [\hat{a}^{\dagger}_{i},\hat{a}^{\dagger}_{j}] = 0$, $[\hat{a}_{i},\hat{a}^{\dagger}_{j}] =\delta_{ij}$, where $\delta_{ij}$ is the Kronecker delta function and $[\hat{A},\hat{B}]\equiv \hat{A}\hat{B}-\hat{B}\hat{A}$. Hilbert space $\mathcal{H}$ is spanned by the eigenstates of the number operator $\hat{n}\equiv \hat{a}^{\dagger}\hat{a}$, i.e., $\mathcal{H}=\textrm{span}\lbrace |n\rangle \rbrace_{n=0}^{\infty}$ where $\hat{n}|n\rangle= n|n\rangle$. In the number basis (or Fock basis), annihilation and creation operators are given by $\hat{a} = \sum_{n=1}^{\infty}\sqrt{n}|n-1\rangle\langle n|$, $\hat{a}^{\dagger} = \sum_{n=0}^{\infty}\sqrt{n+1}|n+1\rangle\langle n|$. Coherent state $|\alpha\rangle$ is an eigenstate of the annihilation operator $\hat{a}$ with eigenvalue $\alpha$, i.e., $\hat{a}|\alpha\rangle = \alpha|\alpha\rangle$. In the Fock basis, $|\alpha\rangle$ is given by $|\alpha\rangle = e^{-\frac{1}{2}|\alpha|^{2}}\sum_{n=0}^{\infty}\frac{\alpha^{n}}{\sqrt{n!}}|n\rangle$. 
Note that the vacuum state $|0\rangle$ is a special case of coherent state with $\alpha=0$. Displacement operator $\hat{D}(\alpha)$ is defined as $\hat{D}(\alpha) \equiv \exp(\alpha\hat{a}^{\dagger}-\alpha^{*}\hat{a})$. Coherent state $|\alpha\rangle$ is a displaced vacuum state: 
\begin{equation}
|\alpha\rangle = \hat{D}(\alpha)|0\rangle.  \label{eq:coherent state equals to displacec vacuum}
\end{equation}

Quadrature operators are defined as $\hat{q}_{k}\equiv \frac{1}{\sqrt{2}}(\hat{a}_{k}+\hat{a}_{k}^{\dagger})$, $\hat{p}_{k}\equiv \frac{i}{\sqrt{2}}(\hat{a}_{k}^{\dagger}-\hat{a}_{k})$ and are called position and momentum operator, respectively. We follow the same convention as used for the GKP codes \cite{Gottesman2001,Albert2017} which differs from Ref. \cite{Weedbrook2012} by a factor of $\sqrt{2}$ in the definition of $\hat{q}_{k}$ and $\hat{p}_{k}$. Define $\mathbf{\hat{x}}\equiv  (\hat{q}_{1},\hat{p}_{1},\cdots,\hat{q}_{N},\hat{p}_{N})^{T}$. 
Then, the bosonic commutation relation reads $[\hat{x}_{i},\hat{x}_{j}]=i \Omega_{ij}$, where $\mathbf{\Omega}$ is defined as 
\begin{equation}
\mathbf{\Omega} \equiv \begin{pmatrix}
\mathbf{\omega} &  &  \\
& \ddots & \\
& & \mathbf{\omega}
\end{pmatrix}  \textrm{ and } \mathbf{\omega} \equiv 
\begin{pmatrix}
0 & 1    \\
-1 & 0
\end{pmatrix}. 
\end{equation}
Eigenvalue spectrum of quadrature operators are continuous, $\hat{q}|q\rangle = q|q\rangle, \quad \hat{p}|p\rangle = p|p\rangle$, where $q,p\in(-\infty,\infty)$ and the eigenstates are normalized by the Dirac delta function: $\langle q|q'\rangle = \delta(q-q')$ and $\langle p|p'\rangle = \delta(p-p')$. Also, $|q\rangle$ and $|p\rangle$ are related by Fourier transformation $|q\rangle = \frac{1}{\sqrt{2\pi}}\int_{-\infty}^{\infty}dp e^{-i q p}|p\rangle$, $|p\rangle = \frac{1}{\sqrt{2\pi}}\int_{-\infty}^{\infty}dp e^{i q p}|q\rangle$. Upon displacement, 
\begin{alignat}{2} 
\hat{D}(\xi_{1}/\sqrt{2})|q\rangle &= e^{-i\xi_{1}\hat{p}}|q\rangle\, &&= |q+\xi_{1}\rangle , 
\nonumber\\
\hat{D}(i\xi_{2}/\sqrt{2})|p\rangle &= e^{i\xi_{2}\hat{q}}|p\rangle &&= |p+\xi_{2}\rangle . 
\end{alignat}

Let $\mathcal{L}(\mathcal{H})$ be the space of linear operators on the Hilbert space $\mathcal{H}$. General quantum states (including mixed states) are described by density operator $\hat{\rho}\in\mathcal{D}(\mathcal{H})$, where $\mathcal{D}(\mathcal{H})\equiv \lbrace \hat{\rho}\in\mathcal{L}(\mathcal{H})\, |\, \hat{\rho}^{\dagger}=\hat{\rho}\succeq 0, \textrm{Tr}[\hat{\rho}]=1\rbrace$. Expectation value of an observable $\hat{E}$ of the state $\hat{\rho}$ is given by $\langle \hat{E} \rangle = \textrm{Tr}[\hat{\rho}\hat{E}]$. Wigner characteristic function $\chi(\mathbf{\xi})$ is defined as $\chi(\mathbf{\xi})\equiv \textrm{Tr}[\hat{\rho}\exp(i\mathbf{\hat{x}}^{T}\mathbf{\Omega} \mathbf{\xi})]$, where $\hat{\rho}\in\mathcal{D}(\mathcal{H}^{\otimes N})$ and $\mathbf{\xi} = (\xi_{1},\cdots,\xi_{2N})^{T}$. Weyl operator $\exp(i\mathbf{\hat{x}}^{T}\mathbf{\Omega} \mathbf{\xi})$ is of the form of displacement operator and satisfies the orthogonality relation $\textrm{Tr}[\exp(-i\mathbf{\hat{x}}^{T}\mathbf{\Omega} \mathbf{\xi})\exp(i\mathbf{\hat{x}}^{T}\mathbf{\Omega} \mathbf{\xi'})] = (2\pi)^{N}\delta(\mathbf{\xi}-\mathbf{\xi'})$. Wigner characteristic function $\chi(\mathbf{\xi})$ is in one-to-one correspondence with the state $\hat{\rho}$ and the inverse function is given by $\hat{\rho} = \frac{1}{(2\pi)^{N}}\int d^{2N}\mathbf{\xi} \chi(\mathbf{\xi})  \exp(-i\mathbf{\hat{x}}^{T}\mathbf{\Omega} \mathbf{\xi})$. Wigner function $W(\mathbf{x})$ is Fourier transformation of $\chi(\mathbf{\xi})$, i.e., 
\begin{equation}
W(\mathbf{x}) = \frac{1}{(2\pi)^{N}}\int d^{2N}\mathbf{\xi} \chi(\mathbf{\xi})  \exp(-i\mathbf{x}^{T}\mathbf{\Omega} \mathbf{\xi}), \label{eq:definition of Wigner function}
\end{equation}  
where $\mathbf{x} = (x_{1},\cdots,x_{N})^{T}$ is the eigenvalue of the quadrature operator $\mathbf{\hat{x}}$. 

A quantum state $\hat{\rho}$ is Gaussian iff its Wigner characteristic function and Wigner function are Gaussian \cite{Weedbrook2012}:
\begin{align} 
\chi(\mathbf{\xi}) &= \exp\big{[} -\frac{1}{2}\mathbf{\xi}^{T}(\mathbf{\Omega} \mathbf{V} \mathbf{\Omega}^{T})\mathbf{\xi} -i (\mathbf{\Omega} \mathbf{\bar{x}})\mathbf{\xi}\big{]},  
\nonumber\\
W(\mathbf{x}) &= \frac{\exp\big{[} -\frac{1}{2}(\mathbf{x}-\mathbf{\bar{x}})^{T}\mathbf{V}^{-1}(\mathbf{x}-\mathbf{\bar{x}})\big{]}}{(2\pi)^{N}\sqrt{\textrm{det}\mathbf{V}}}, \label{eq:Wigner characteristic function for Gaussian states}
\end{align}
where $\mathbf{\bar{x}},\mathbf{V}$ are first and second moments of the state $\hat{\rho}$. 
\begin{equation}
\mathbf{\bar{x}} \equiv  \langle \mathbf{\hat{x}} \rangle  = \textrm{Tr}[\hat{\rho}\mathbf{\hat{x}}], \quad V_{ij} \equiv \frac{1}{2}\langle \lbrace \hat{x}_{i}-\bar{x}_{i} , \hat{x}_{j}-\bar{x}_{j}  \rbrace \rangle , 
\end{equation} 
where $\lbrace \hat{A},\hat{B}\rbrace\equiv \hat{A}\hat{B}+\hat{B}\hat{A}$. Thus, Gaussian states are fully characterized by the first two moments, i.e., $\hat{\rho}=\hat{\rho}_{G}(\mathbf{\bar{x}},\mathbf{V})$. Heisenberg uncertainty relation reads $\mathbf{V}+\frac{i}{2}\mathbf{\Omega} \succeq 0$ and implies $V(\hat{q}_{k})V(\hat{p}_{k})\ge \frac{1}{4}$ for all $k\in\lbrace 1,\cdots,N \rbrace$, where $V(\hat{x}_{i}) \equiv \mathbf{V}_{ii}$. 

Vacuum state $|0\rangle\langle 0|$ is the simplest example of one-mode Gaussian states with $\mathbf{\bar{x}}=\mathbf{0}$ and $\mathbf{V} = \frac{\mathbf{I}_{2}}{2}$, where $\mathbf{I}_{n}$ is defined as the $n\times n$ identity matrix. Coherent state $|\alpha\rangle\langle\alpha|$ is also Gaussian: $|\alpha\rangle\langle\alpha| = \hat{\rho}_{G}(\mathbf{\bar{x}}_{\alpha},\frac{\mathbf{I}_{2}}{2})$ with $\mathbf{\bar{x}}_{\alpha}\equiv \sqrt{2}(\alpha_{R},\alpha_{I})^{T}$ and $\alpha= \alpha_{R}+i\alpha_{I}$. Coherent states (including the vacuum state) saturate the uncertainty inequality and thus have the minimum uncertainty. Thermal state is an example of Gaussian mixed states and is given by
\begin{equation}
\hat{\rho}_{\nthtiny} \equiv \sum_{n=0}^{\infty} \frac{(\nth)^{n}}{(\nth+1)^{n+1}} |n\rangle\langle n| =\hat{\rho}_{G}(\mathbf{0},(\nth+\frac{1}{2})\mathbf{I}_{2}) ,   \label{eq:thermal state definition and Gaussian specification} 
\end{equation} 
in the Fock basis, where $\nth$ is the average photon number, i.e., $\nth = \textrm{Tr}[\hat{\rho}_{\nthtiny}\hat{n}]$. Quantum entropy of a state $\hat{\rho}$ is defined as $H(\hat{\rho})\equiv -\textrm{Tr}[\hat{\rho}\log\hat{\rho}]$, where $\log$ is the logarithm with base $2$. Entropy of a thermal state $\hat{\rho}_{\nthtiny}$ is given by 
\begin{equation}
H(\hat{\rho}_{\nthtiny}) = g(\nth), \label{eq:thermal state entropy}  
\end{equation}
where $g(x) \equiv (x +1) \log(x+1)-x\log x$. Since thermal state is a mixed state, $H(\hat{\rho}_{\nthtiny})\ge 0$ where the equality holds only when the state is vacuum, i.e., $\nth=0$.  

Unitary operations that map a Gaussian state to another Gaussian state are called Gaussian unitary operations. Gaussian unitaries are generated by second order polynomials of $\mathbf{\hat{a}} = (\hat{a}_{1},\cdots,\hat{a}_{N})^{T} $ and $\mathbf{\hat{a}^{\dagger}} = (\hat{a}_{1}^{\dagger},\cdots,\hat{a}_{N}^{\dagger})^{T}$, i.e., $\hat{U}_{G} = \exp(-i\hat{H})$ with $\hat{H} = i (\mathbf{\alpha}^{T}\mathbf{\hat{a}^{\dagger}} + \mathbf{\hat{a}^{\dagger}}\mathbf{F}\mathbf{\hat{a}} + \mathbf{\hat{a}^{\dagger}}\mathbf{G}\mathbf{\hat{a}^{\dagger}}) + \textrm{h.c.}$, where $\mathbf{\alpha}^{T} = (\alpha_{1},\cdots,\alpha_{N})$ and $\mathbf{F},\mathbf{G}$ are $N\times N$ complex matrices. In the Heisenberg picture, annihilation operator $\mathbf{\hat{a}}$ is transformed into $\hat{U}_{G}^{\dagger} \mathbf{\hat{a}} \hat{U}_{G}  = \mathbf{A}\mathbf{\hat{a}} + \mathbf{B}\mathbf{\hat{a}^{\dagger}} +\mathbf{\alpha}$, where $N\times N$ complex matrices $\mathbf{A},\mathbf{B}$ (determined by $\mathbf{F},\mathbf{G}$) satisfy $\mathbf{A}\mathbf{B}^{T} = \mathbf{B}\mathbf{A}^{T}$ and $\mathbf{A}\mathbf{A}^{\dagger}=\mathbf{B}\mathbf{B}^{\dagger}+\mathbf{I}_{N}$. In terms of quadrature operators, the transformation reads $\mathbf{x}\rightarrow \hat{U}_{G}^{\dagger}\mathbf{x}\hat{U}_{G} = \mathbf{S}\mathbf{x} + \mathbf{d}$, where $\mathbf{d} = (d_{1},\cdots,d_{2N})^{T}=\sqrt{2}(\alpha_{1}^{R},\alpha_{1}^{I},\cdots,\alpha_{N}^{R},\alpha_{N}^{I})^{T}$ and $2N\times 2N$ matrix $\mathbf{S}$ is symplectic: 
\begin{equation}
\mathbf{S}\mathbf{\Omega}\mathbf{S}^{T}  =\mathbf{\Omega}. \label{eq:symplectic matrix defining property}
\end{equation}
Gaussian unitaries are thus fully characterized by $\mathbf{S},\mathbf{d}$, and under $\hat{U}_{\mathbf{S},\mathbf{d}}$ the first two moments are transformed as 
\begin{equation}
\mathbf{\bar{x}} \rightarrow \mathbf{S}\mathbf{\bar{x}}+\mathbf{d}, \quad \mathbf{V}\rightarrow\mathbf{S}\mathbf{V}\mathbf{S}^{T}.  \label{eq:transformation of moments of quadrature by Gaussian unitary}
\end{equation}

Displacement operator $\hat{D}(\alpha)$ is a one-mode Gaussian unitary operation with $\alpha$ and $\mathbf{F}=\mathbf{G}=0$, yielding $\mathbf{A}=1,\mathbf{B}=0$ and $\mathbf{S}=\mathbf{I}_{2}$, $\mathbf{d}=\sqrt{2}(\alpha_{R},\alpha_{I})^{T}$. Squeeze operator $\hat{S}(r)\equiv \exp(\frac{r}{2}(\hat{a}^{2}-\hat{a}^{\dagger 2}) )$ is another example of one-mode Gaussian unitaries and transforms quadrature operators by $\hat{q}\rightarrow e^{-r}\hat{q}$ and $\hat{p}\rightarrow e^{r}\hat{p}$, i.e., $\mathbf{S}=\textrm{diag}(e^{-r},e^{r})$. Quadrature eigenstate can thus be understood as an infinitely squeezed state: For example, $|\hat{q}=0\rangle \propto \lim_{r \rightarrow+\infty} \hat{S}(r)|0\rangle$ and $|\hat{p}=0\rangle \propto \lim_{r \rightarrow-\infty} \hat{S}(r)|0\rangle$, where $|0\rangle$ is the vacuum state. Phase rotation operator is defined as $\hat{U}(\theta)\equiv \exp(i\theta\hat{a}^{\dagger}\hat{a})$. Under the phase rotation, quadrature operators are transformed as 
\begin{equation}
\mathbf{x} \rightarrow \mathbf{R}(\theta)\mathbf{x} \,\,\, \textrm{and} \,\,\,  \mathbf{R}(\theta) \equiv \begin{pmatrix}
\cos\theta & -\sin\theta    \\
\sin\theta & \cos\theta
\end{pmatrix}, \label{eq:symplectic transformation phase rotation}
\end{equation}
yielding, e.g., $\hat{U}(\theta)|\alpha\rangle = |\alpha e^{i\theta}\rangle$. 

Beam splitter is a two-mode Gaussian unitary operation generated by the Hamiltonian of the form $\hat{H}\propto i(\hat{a}_{1}^{\dagger}\hat{a}_{2}-\hat{a}_{1}\hat{a}_{2}^{\dagger})$.
Beam splitter unitary $\hat{B}(\eta)$ transforms the annihilation operators by $\hat{a}_{1}\rightarrow \sqrt{\eta}\hat{a}_{1}+\sqrt{1-\eta}\hat{a}_{2}$ and $\hat{a}_{2}\rightarrow -\sqrt{1-\eta}\hat{a}_{1}+\sqrt{\eta}\hat{a}_{2}$, where $\eta=\cos^{2}\theta\in[0,1]$ is called the transmissivity. In terms of the quadrature operator $\mathbf{\hat{x}}=(\hat{q}_{1},\hat{p}_{1},\hat{q}_{2},\hat{p}_{2})^{T}$, the transformation reads
\begin{equation}
\mathbf{\hat{x}}\rightarrow \mathbf{B}(\eta) \mathbf{\hat{x}} \,\,\textrm{and}\,\,  \mathbf{B}(\eta)\equiv \begin{pmatrix}
\sqrt{\eta}\mathbf{I}_{2} & \sqrt{1-\eta}\mathbf{I}_{2}   \\
-\sqrt{1-\eta}\mathbf{I}_{2} & \sqrt{\eta}\mathbf{I}_{2}
\end{pmatrix}. \label{eq:symplectic transformation beam splitter}
\end{equation}
Another example of two-mode Gaussian operation is the two-mode squeezing generated by $\hat{H}\propto i(\hat{a}_{1}\hat{a}_{2}-\hat{a}_{1}^{\dagger}\hat{a}_{2}^{\dagger})$. Under the two-mode squeezing $\hat{S}_{2}(G)$, annihilation operators are transformed as $\hat{a}_{1}\rightarrow \sqrt{G}\hat{a}_{1}+\sqrt{G-1}\hat{a}_{2}^{\dagger}$ and $\hat{a}_{2}\rightarrow \sqrt{G-1}\hat{a}_{1}^{\dagger}+\sqrt{G}\hat{a}_{2}$, where $G\ge 1$ is the gain. Under $\hat{S}_{2}(G)$, quadrature operators are transformed as 
\begin{equation}
\mathbf{\hat{x}}\rightarrow \mathbf{S}_{2}(G) \mathbf{\hat{x}} \,\,\textrm{and}\,\,  \mathbf{S}_{2}(G)\equiv \begin{pmatrix}
\sqrt{G}\mathbf{I}_{2} & \sqrt{G-1}\mathbf{Z}_{2}   \\
\sqrt{G-1}\mathbf{Z}_{2} & \sqrt{G}\mathbf{I}_{2}
\end{pmatrix},  \label{eq:symplectic transformation two mode squeezer}
\end{equation}
where $\mathbf{Z}_{2}\equiv \textrm{diag}(1,-1)$. 

\section{Characterization of Gaussian random displacement channel}
\label{subsection:Characterization of Gaussian random displacement channel}

Gaussian random displacement channel $\mathcal{N}_{B_{2}}[\sigma^{2}](\hat{\rho})$ as defined in Definition \ref{definition:Gaussian displacement channel} is characterized by $\mathbf{T}=\mathbf{I}_{2}$, $\mathbf{N}=\sigma^{2}\mathbf{I}_{2}$ and $\mathbf{d}=0$. 
\begin{proof}
Let $\hat{\rho}'\equiv \mathcal{N}_{B_{2}}[\sigma^{2}](\hat{\rho})$. Note that, in the case of one-mode, $\exp(i\mathbf{\hat{x}}^{T}\mathbf{\Omega}\mathbf{\xi}) = \hat{D}(\xi/\sqrt{2})$ where $\xi = \xi_{1}+i\xi_{2}$ (see Appendix \ref{section:Bosonic mode, Gaussian states and Gaussian unitary operations} for the definition of $\mathbf{\Omega}$ and $\mathbf{\hat{x}}$). The Wigner characteristic function of $\hat{\rho}'$ is then given by 
\begin{align} 
\chi'(\mathbf{\xi}) &= \textrm{Tr}\Big{[}\mathcal{N}_{B_{2}}[\sigma^{2}](\hat{\rho})\hat{D}\Big{(}\frac{\xi}{\sqrt{2}}\Big{)}\Big{]}
\nonumber\\
&=  \textrm{Tr}\Big{[}\frac{1}{\pi\sigma^{2}} \int d^{2}\alpha e^{-\frac{|\alpha|^{2}}{\sigma^{2}}}  \hat{\rho}  \hat{D}^{\dagger}(\alpha)  \hat{D}\Big{(}\frac{\xi}{\sqrt{2}}\Big{)}\hat{D}(\alpha) \Big{]}
\nonumber\\
&= \frac{1}{\pi\sigma^{2}} \int d^{2}\alpha e^{-\frac{|\alpha|^{2}}{\sigma^{2}}} e^{\frac{1}{\sqrt{2}}(\alpha^{*}\xi-\alpha\xi^{*})}  \textrm{Tr}\Big{[} \hat{\rho}   \hat{D}\Big{(}\frac{\xi}{\sqrt{2}}\Big{)} \Big{]}
\nonumber\\
&= e^{-\frac{\sigma^{2}}{2}|\xi|^{2}}\chi(\mathbf{\xi}), 
\end{align}
where $\chi(\xi)$ is the Wigner characteristic function of the input state $\hat{\rho}$--cf, text above Eq.\eqref{eq:definition of Wigner function}. We used cyclic property of trace to derive the second equality, and applied Baker--Campbell--Hausdorff formula to obtain $\hat{D}^{\dagger}(\alpha)\hat{D}(\beta)\hat{D}(\alpha) = \exp(\alpha^{*}\beta-\alpha\beta^{*})\hat{D}(\beta)$ and the third equality. The last equality follows from the evaluation of Gaussian integral and the definition of $\chi(\mathbf{\xi})$. If the input state is Gaussian, i.e., $\hat{\rho}=\hat{\rho}_{G}(\mathbf{\bar{x}},\mathbf{V})$, Winger characteristic function of the output state is given by
\begin{equation}
\chi'(\mathbf{\xi}) = \exp\big{[} -\frac{1}{2}\mathbf{\xi}^{T}(\mathbf{\Omega}  (\mathbf{V}+\sigma^{2}\mathbf{I}_{2}) \mathbf{\Omega}^{T})\mathbf{\xi} -i (\mathbf{\Omega} \mathbf{\bar{x}})\mathbf{\xi}\big{]},  
\end{equation}
and thus $\hat{\rho}'=\hat{\rho}_{G}(\mathbf{\bar{x}},\mathbf{V}+\sigma^{2}\mathbf{I}_{2}$).   
\end{proof}  

\section{Derivation of the optimized data-processing bound}
\label{section:Derivation of the optimized data-processing bound} 

\begin{proof}
Assume that a Gaussian loss channel is decomposed as 
\begin{equation}
\mathcal{N}[\eta,\nth] =\mathcal{A}[G_{1}] \mathcal{N}[\bar{\eta},0]\mathcal{A}[G_{2}], 
\end{equation}
cf, Eq.\eqref{eq:decomposition of thermal loss into preamp pure loss and post amp}. Then, the channel on the right hand side is characterized by $\mathbf{T} = \sqrt{G_{1}\bar{\eta}G_{2}}\mathbf{I}_{2}$ and 
\begin{align}
\mathbf{N} &= \Big{[} G_{1}\Big{(}\bar{\eta}\times \frac{1}{2}(G_{2}-1) + \frac{1}{2} (1-\bar{\eta}) \Big{)} + \frac{1}{2}(G_{1}-1)  \Big{]} \mathbf{I_{2}} 
\nonumber\\
&= \frac{1}{2}\Big{[} G_{1}\bar{\eta}G_{2} -2G_{1}\bar{\eta} +2G_{1}-1 \Big{]}\mathbf{I_{2}} . 
\end{align}
Imposing $\mathbf{T} = \sqrt{\eta}\mathbf{I}_{2}$ and $\mathbf{N}=(1-\eta)(\nth+\frac{1}{2})\mathbf{I}_{2}$, we get $(1-\bar{\eta})G_{1} = (1-\eta)(\nth+1)$ and thus $\bar{\eta} = 1-\frac{1}{G_{1}} (1-\eta)(\nth+1)$, $G_{2} = \frac{\eta}{G_{1}\bar{\eta}} = \frac{\eta}{G_{1}-(1-\eta)(\nth+1)}$. Then, $Q_{\reg}^{n\le \bar{n}}(\mathcal{N}[\eta,\nth])$ is upper bounded by 
\begin{align}
&Q^{n\le G_{2}\bar{n}+(G_{2}-1)}(\mathcal{N}[\bar{\eta},0])
\nonumber\\
&\quad\equiv \max\Big{[}  g\big{(}\bar{\eta}(G_{2}\bar{n}+(G_{2}-1))\big{)}
\nonumber\\
&\qquad\qquad\qquad- g\big{(}(1-\bar{\eta})(G_{2}\bar{n}+(G_{2}-1))\big{)},0 \Big{]} .
\end{align}
hence Eq.\eqref{eq:definition of the f function}. To ensure $G_{1}\ge 1$ and $G_{2} = \frac{\eta}{G_{1}-(1-\eta)(\nth+1)} \ge 1$, $G_{1}$ should be in the range $1 \le G_{1}\le 1+(1-\eta)\nth$. Thus, Eq.\eqref{eq:optimized data-processing bound} follows.   
\end{proof}

\bibliography{GKP_v6_arxiv}







\end{document}